\documentclass[a4paper,12pt,fleqn]{article}

\usepackage{amsmath,amssymb,amsthm}
\usepackage{thm-restate}
\usepackage[margin=1in]{geometry}
\usepackage{xcolor}
\usepackage{url}
\usepackage{array}
\usepackage{comment}
\usepackage[hypertexnames=false,bookmarksnumbered=true,final]{hyperref}
\usepackage[capitalize,sort]{cleveref}

\def\colorschemesepia{sepia}
\def\colorschemedark{dark}
\def\colorschemelight{light}

\ifx\colorscheme\undefined
\let\colorscheme\colorschemelight
\fi

\ifx\colorscheme\colorschemelight
\colorlet{textColor}{black}
\colorlet{bgColor}{white}
\fi

\ifx\colorscheme\colorschemesepia
\definecolor{textColor}{HTML}{433423}
\definecolor{bgColor}{HTML}{fbf0da}
\fi

\ifx\colorscheme\colorschemedark
\definecolor{textColor}{HTML}{bdc1c6}
\definecolor{bgColor}{HTML}{202124}
\definecolor{textBlue}{HTML}{8ab4f8}
\definecolor{textRed}{HTML}{f9968b}
\definecolor{textGreen}{HTML}{81e681}
\definecolor{textPurple}{HTML}{c58af9}
\else
\colorlet{textBlue}{blue!50!black}
\colorlet{textRed}{red!50!black}
\colorlet{textGreen}{green!50!black}
\definecolor{textPurple}{HTML}{681da8}
\fi

\ifx\colorscheme\colorschemelight\else
\pagecolor{bgColor}
\color{textColor}
\fi

\hypersetup{colorlinks,linkcolor=textRed,citecolor=textRed,urlcolor=textBlue}

\let\eps\varepsilon

\newcommand*{\defeq}{:=}
\newcommand*{\Th}{^{\textrm{th}}}

\newcommand*{\ceil}[1]{\left\lceil #1 \right\rceil}
\DeclareMathOperator*{\E}{\mathbb{E}}
\DeclareMathOperator*{\Var}{Var}

\newcommand*{\WLoG}{Without loss of generality}
\newcommand*{\wLoG}{without loss of generality}
\newcommand*{\toN}[1][n]{[1..#1]}
\newcommand*{\trn}{^T\mspace{-2.5mu}}
\newcommand*{\trnF}{^T}

\newcommand*{\boolone}{\mathbf{1}}
\newcommand*{\simplex}{\Delta}
\newcommand*{\sampledFrom}{\sim}
\DeclareMathOperator{\support}{support}
\DeclareMathOperator{\graph}{graph}

\DeclareMathOperator{\colls}{collisions}
\DeclareMathOperator{\isResponse}{isResp}
\DeclareMathOperator{\stdmeml}{stat}
\DeclareMathOperator{\perronSolution}{perronSolution}
\DeclareMathOperator{\concat}{concat}

\newcommand*{\batter}{max player}
\newcommand*{\bowler}{min player}
\newcommand*{\batterAndBowler}{max and min players}

\let\mymathvec\mathrm
\newcommand*{\zerovec}{\mymathvec{0}}

\newcommand*{\evec}{\mymathvec{e}}

\newcommand*{\uvec}{\mymathvec{u}}
\newcommand*{\vvec}{\mymathvec{v}}
\newcommand*{\wvec}{\mymathvec{w}}
\newcommand*{\xvec}{\mymathvec{x}}
\newcommand*{\yvec}{\mymathvec{y}}
\newcommand*{\zvec}{\mymathvec{z}}

\newcommand*{\xstar}{\xvec^*}
\newcommand*{\xstarT}{{\xvec^*}\trn}
\newcommand*{\ystar}{\yvec^*}
\newcommand*{\xhat}{\widehat{\xvec}}
\newcommand*{\yhat}{\widehat{\yvec}}
\newcommand*{\xtild}{\widetilde{\xvec}}
\newcommand*{\ytild}{\widetilde{\yvec}}
\newcommand*{\fxvec}{f_{\xvec}}
\newcommand*{\gyvec}{g_{\yvec}}

\newcommand*{\fhat}{\widehat{f}}
\newcommand*{\ftild}{\widetilde{f}}
\newcommand*{\ghat}{\widehat{g}}
\newcommand*{\gtild}{\widetilde{g}}
\newcommand*{\ihat}{\hat\imath}
\newcommand*{\jhat}{\hat\jmath}
\newcommand*{\rhohat}{\widehat{\rho}}

\newcommand*{\Ebar}{\overline{E}}
\newcommand*{\Dcal}{\mathcal{D}}
\newcommand*{\Ecal}{\mathcal{E}}
\newcommand*{\Fcal}{\mathcal{F}}
\newcommand*{\Gcal}{\mathcal{G}}

\newcommand*{\Xcal}{\mathcal{X}}
\newcommand*{\Ycal}{\mathcal{Y}}

\newcommand*{\movedProof}[1]{\begin{proof}(See #1.)\end{proof}}

\newtheorem{theorem}{Theorem}

\newtheorem{definition}{Definition}
\newtheorem{example}{Example}

\newtheorem{lemma}[theorem]{Lemma}
\newtheorem{observation}[theorem]{Observation}
\newtheorem{proposition}[theorem]{Proposition}

\allowdisplaybreaks
\makeatletter
\g@addto@macro{\UrlBreaks}{%
\do\/%
\do\a\do\b\do\c\do\d\do\e\do\f\do\g\do\h\do\i\do\j\do\k\do\l\do\m%
\do\n\do\o\do\p\do\q\do\r\do\s\do\t\do\u\do\v\do\w\do\x\do\y\do\z%
\do\A\do\B\do\C\do\D\do\E\do\F\do\G\do\H\do\I\do\J\do\K\do\L\do\M%
\do\N\do\O\do\P\do\Q\do\R\do\S\do\T\do\U\do\V\do\W\do\X\do\Y\do\Z%
\do\0\do\1\do\2\do\3\do\4\do\5\do\6\do\7\do\8\do\9%
}
\makeatother

\newcommand*{\acknowledgements}[1]{\subparagraph{Acknowledgements.} #1}

\title{Nash Equilibria of Two-Player Matrix Games Repeated Until Collision}

\newcommand*{\emailfont}[1]{\texttt{\small #1}}

\author{Aniket Murhekar%
\thanks{University of Illinois, Urbana-Champaign, USA}%
\\\emailfont{aniket2@illinois.edu}
\and Eklavya Sharma\footnotemark[1]\\\emailfont{eklavya2@illinois.edu}}

\date{\empty}

\begin{document}

\maketitle
\setlength{\parskip}{0.2em}

\begin{abstract}
We introduce and initiate the study of a natural class of repeated two-player matrix games, called \emph{Repeated-Until-Collision} (RUC) games.
In each round, both players simultaneously pick an action from a common action set $\{1, 2, \dots, n\}$.
Depending on their chosen actions, they derive payoffs given by $n \times n$ matrices $A$ and $B$, respectively.
If their actions \emph{collide} (i.e., they pick the same action), the game ends, otherwise, it proceeds to the next round.
Both players want to maximize their total payoff until the game ends.
RUC games can be interpreted as pursuit-evasion games or repeated hide-and-seek games.
They also generalize \emph{hand cricket}, a popular game among children in India.

We show that under mild assumptions on the payoff matrices, every RUC game admits a Nash equilibrium (NE).
Moreover, we show the existence of a \emph{stationary} NE, where each player chooses their action according to a probability distribution over the action set that does not change across rounds.
Remarkably, we show that all NE are effectively the same as the stationary NE, thus showing that RUC games admit an \emph{almost} unique NE.
Lastly, we also show how to compute (approximate) NE for RUC games.

\end{abstract}

\section{Introduction}

Two-player matrix games, or \textit{bimatrix} games, are among the most well-studied classes of games in non-cooperative game theory \cite{osborneGameTheoryTextBook}. A bimatrix game involves two players with a finite set of actions $\{1,2,\dots,n\}$, and can be represented by two $n\times n$ payoff matrices $A$ and $B$ for the two players. A play of the game can be thought of as one player choosing a row $i$ and the other player choosing a column $j$ simultaneously. The `row' player gets a payoff of $A_{i,j}$ and the `column' player gets a payoff of $B_{i,j}$. To avoid being predictable, players can choose their actions according to a \textit{randomized} or mixed strategy which is a probability distribution over their action set. Naturally, each player is interested in maximizing their (expected) payoff.

Arguably the most popular solution concept in game theory is that of Nash equilibrium \cite{nash1951NE}. A Nash equilibrium (NE) for a bimatrix game is a pair of mixed strategies where no player has any incentive to unilaterally deviate and change her strategy. In a celebrated result, Nash showed that each bimatrix game admits a Nash equilibrium. Given that two-player games are ubiquitous and have widespread applications in networks and communications \cite{gtappwireless}, financial markets \cite{gtappfinance}, robotics \cite{gtapprobots}, etc., the field of algorithmic game theory has extensively studied the existence and computation of NE in bimatrix games and their generalizations.

A repeated two-player game is one such generalization, which involves repetitions of a base bimatrix game \cite{repeatedGamesBook,zerosumrepeatedgames,finitelyRepeatedGames}. In a repeated game, a player can choose their moves based on the moves played by the players in the previous rounds, while recognizing that their current move will impact the choice of moves of their opponent in the future, leading to added complexity. While finitely repeated games are well-understood using backward-induction \cite{finitelyRepeatedGames}, the equilibria of infinitely repeated games can differ from that of the base game \cite{repeatedGamesBook}.

In this paper, we introduce and initiate the study of a natural class of repeated bimatrix games, called \emph{Repeated-Until-Collision} (RUC) games.
For $n\times n$ non-negative matrices $A$ and $B$, an RUC game $(A,B)$ has two players, called \emph{\batter} and \emph{\bowler}, where $A$ is called the \batter's score matrix and $B$ is called the \bowler's cost matrix. The game consists of multiple rounds.
Suppose the \batter{} and \bowler{} pick actions $i$ and $j$, respectively, in a round.
Then the \batter{} earns a score of $A_{i,j}$ and the \bowler{} incurs a cost of $B_{i,j}$.
If $i = j$, the game ends (we call this event a \emph{collision}). Otherwise, the game proceeds to the next round. We assume that the \batter{} and \bowler{} may use randomized strategies.
The \batter{} wants to maximize her expected total score and the \bowler{} wants to minimize her expected total cost. Below we discuss a few applications of RUC games, thus underscoring their importance.

\subparagraph{Applications.} RUC games can be interpreted as variants of pursuit-evasion games \cite{khan2007game,borie09pursuitevasion}. As a simplified example, consider a game between a drug dealer (\batter{}) and the law enforcement (\bowler{}). Each day, the drug dealer chooses one of $n$ locations for a drop, while law enforcement chooses a location for a random check. If their locations coincide, the game ends as the drug dealer is caught. Until then, the drug dealer wants to maximize his revenue from the sale of drugs, while law enforcement wants to minimize the cost to society by illicit drug use. Similar examples can be found in reinforcement learning or robot motion planning, where an agent, e.g., a spy robot, is trying to learn an unknown environment. The agent gets a reward from exploring one of $n$ locations, while its adversary (security systems, nature, etc.) actively tries to minimize how much the agent discovers before catching the infiltrating agent.

RUC games generalize \textit{hand cricket} \cite{hc-instructables,hc-ios-app}, a popular game played by children in India. Hand cricket is a contest between a `batter' (\batter{}) and a `bowler' (\bowler{}). In a round, the batter and bowler simultaneously choose actions $i$ and $j$ from $\{1,\dots,n\}$. This gives the batter a score of $i$, while the bowler suffers a cost of $i$. The game ends if $i=j$, i.e., the batter is declared `out'. Hand cricket involves the batter trying to maximize her total score and the bowler trying to minimize it. Note that this is an example of a \textit{zero-sum} RUC game, where the payoff to the \batter{} equals the cost borne by the \bowler{}.

We also note that RUC games can be thought of as repeated hide-and-seek games \cite{Rubinstein1997hideseek,Rubinstein1993hideseek} between two players, the Hider and the Seeker. In a hide-and-seek game, there are $n$ locations containing varying rewards. The Hider tries to collect as much reward as possible, before getting caught by the Seeker, who aims to minimize the reward lost to the Hider. Like hand cricket, repeated hide-and-seek is also a zero-sum RUC game.

\subparagraph{Nash equilibria in RUC games.} Our paper addresses the following natural questions about RUC games:
\begin{quote}
\textit{Do Repeated-Until-Collision games admit Nash equilibria? If so, are they unique? Can they be computed efficiently?}
\end{quote}

At first glance, it is not even clear if an NE should exist for general RUC games. Intuitively, the \bowler{} would like a collision to happen soon to prevent accumulating a large cost, while the \batter{} would like to delay a collision.
Additionally, if some row of $A$ has large numbers, the \batter{} would want to pick the action corresponding to that row more frequently.
These two approaches are at odds with each other, e.g., if the \batter{} picks a single action
very frequently, the \bowler{} can pick the same action and cause a collision very soon.
Hence, the players must pick distributions of actions that balance
the per-round score (or cost) and the duration of the game. Whether distributions in which neither player has an incentive to deviate unilaterally (i.e., Nash equilibrium) exist is unclear.

\subsection{Our Results}

As a warm-up, we begin by looking at a simpler version of RUC games called \emph{stationary Repeated-Until-Collision} (SRUC) games, where the players are restricted to use only \textit{stationary strategies}. A stationary strategy is one where in each round, the player samples an action from the same distribution.
Formally, given a vector $\xvec\in\simplex_n$ (where $\simplex_n$ is the $n$-dimensional standard simplex),
the stationary strategy $\xvec$ is to pick action $i$
in each round independently randomly with probability $\xvec_i$ for each $i \in \{1,\ldots,n\}$.
Intuitively, stationary strategies make sense for RUC games, since the game has a
recursive structure: if a collision doesn't happen in the first round,
the remaining game is identical to the original. Moreover, stationary strategies are a natural choice for players constrained on computational resources, since they are independent of the game history.

Under mild assumptions on matrices $A$ and $B$, we show that a Nash equilibrium
always exists for the SRUC game $(A, B)$, and is unique.

\begin{definition}
\label{defn:irred}
For any matrix $A \in \mathbb{R}_{\ge 0}^{n \times n}$, let
$\graph(A) \defeq (V, E)$ be a directed graph where
$V \defeq \{1, \ldots, n\}$ and
$E \defeq \{(i, j): i \in V, j \in V \setminus \{i\}, A_{i,j} > 0\}$.
Then $A$ is called \emph{irreducible} iff $\graph(A)$ is strongly connected.
\end{definition}

\begin{theorem}
\label{thm:meml-ne-ex-and-uniq}
Let $A$ and $B$ be irreducible matrices.
Then a Nash equilibrium exists for the SRUC game $(A, B)$.
Furthermore, the Nash equilibrium is unique iff $\graph(A)$ is a subgraph of $\graph(B)$.
\end{theorem}

We prove \cref{thm:meml-ne-ex-and-uniq} in \cref{sec:meml}. Our existence result uses the Perron-Frobenius theorem \cite{perronfrobenius}, a central result from matrix theory, and shows that the NE strategies can be computed from the leading eigenvectors of the payoff matrices.

In \cref{sec:reducible}, using more involved techniques,
we also show the existence of NE for SRUC games with reducible payoff matrices.
Furthermore, if either $A$ or $B$ is reducible, Nash equilibria need not be unique.

In \cref{sec:non-meml}, we switch back from SRUC games to RUC games,
i.e., we allow players to play non-stationary strategies.
This setting is significantly more challenging to analyze due to its large and complicated strategy space.
Since strategies can be history-dependent, a deterministic strategy can be viewed as a function
assigning an action to every possible state of the game, where the state is given
by the players' actions in past rounds.
Such functions have an infinite domain since the game's history can grow arbitrarily large, which
implies that the set of deterministic strategies is uncountably infinite.
So, randomized strategies, defined as probability distributions over deterministic strategies,
are tough to analyze formally since these probability distributions have a large and unusual sample space.

Although RUC games allow non-stationary strategies, it is not obvious if players benefit from this extra freedom. If one player uses a stationary strategy, it is a priori unclear if the other player can gain by deviating unilaterally to a non-stationary strategy. Interestingly, we show that every RUC game admits a Nash equilibrium with stationary strategies.

\begin{theorem}
\label{thm:non-meml-ne-ex-and-uniq}
Let $A$ and $B$ be irreducible matrices.
Then there exists a pair of stationary strategies $(\xstar, \ystar)$ that is
a Nash equilibrium for both the RUC game $(A, B)$ and the SRUC game $(A, B)$.
\end{theorem}

Furthermore, we show that for zero-sum RUC games (i.e., when $A = B$), the Nash equilibrium $(\xstar, \ystar)$
is \emph{unique up to equivalence}, i.e., it is impossible to distinguish between
different Nash equilibria just by observing the players' actions.

Having studied the existence of Nash equilibria in RUC games, we now turn to computation. Since the eigenvectors could be irrational, we consider \textit{approximate} Nash equilibria. In an $\eps$-approximate NE, neither player can improve their payoff by a factor of $(1+\eps)$ through unilateral deviations. We show that fine-enough approximations to the leading eigenvectors of the payoff matrices can be used to compute approximate NE of an RUC game.

\begin{theorem}
\label{thm:approx-ne-informal}
Let $(\xstar, \ystar)$ be a Nash equilibrium for the RUC game $(A, B)$ with full support. For $\eps\in[0,1)$, let $\xhat, \yhat \in \simplex_n$ be such that $|\xstar_i-\xhat_i| \le \eps\xstar_i$ and $|\ystar_i-\yhat_i| \le \eps\ystar_i$, for all $i\in \toN$.
Then $(\xhat, \yhat)$ is a $\frac{4\eps}{(1-\eps)^2}$-approximate Nash equilibrium.
\end{theorem}

Estimating the leading eigenvalue and eigenvector of a matrix can be done via the power method \cite{poweriteration}, an iterative method with a linear convergence rate.

\subsection{Related Work}

Among the earliest results on bimatrix games are von Neumann's Minimax Theorem in 1928 on zero-sum games \cite{vonNeumann1928}, and Nash's fundamental result proving the existence of Nash equilibria in bimatrix games \cite{nash1951NE}. Subsequently the field of algorithmic game theory \cite{agtbook} has devoted considerable attention to the computation of NE in games, with a series of works showing PPAD-completeness \cite{papa94PPAD} even for computing an approximate NE
\cite{papa94PPAD, dask05Threenash, chen06NEapproxsmooth}. PPAD-hardness was then shown for many important subclasses, including constant-rank games \cite{mehta14const}, sparse games \cite{Chen06Sparse}, win-lose games \cite{chen07approxwinlose}, etc. On the other hand, approximation schemes for finding approximate NE are known for classes like low rank games \cite{kannanTheobald, AlonVempala} (FPTAS), and when $(A+B)$ is sparse \cite{Barman15} (PTAS).

Repeated games are well-understood in terms of `Folk Theorems', which indicate that several models of repeated games admit many Nash equilibria \cite{friedman1971}. There are many models of repeated games, differing on the horizon for which the game is played (finite or infinite), aggregate utility to the players (arithmetic mean of payoffs in each round, or a sum of payoffs with a discount $\delta\in(0,1)$), and the kind of equilibrium in consideration (Nash or subgame perfect Nash equilibrium). Finite games with arithmetic mean of payoffs as the utility admit a Nash equilibrium via backward induction \cite{finitelyRepeatedGames}.
Infinite games where the utility is the limit of the arithmetic mean of the payoffs have been shown to admit Nash equilibria assuming certain kinds of punishments used to deter players from deviating \cite{Aumann1994,rubinsteininfinitegames}. Infinite games with the utility being a discounted sum of payoffs per round have also been shown to admit equilibria under different conditions on the base game and punishments \cite{fudenbergmaskin,friedman1971}.
In contrast to these works, in our model of RUC games, the utility of a player is the (undiscounted) sum of per-round payoffs. Furthermore, while an RUC game can allow for infinitely many rounds, under reasonable conditions on the payoff matrices (discussed in \cref{sec:meml}), the play terminates in finite time with probability 1 due to collisions. We also do not assume any external model of punishment to prevent agents from deviating.

Finally, we note a superficial similarity of RUC games with stochastic games \cite{shapley1953stochastic}. A stochastic game is a repeated game with an underlying state space. In each round, players simultaneously choose actions from some action set, based on which they get payoffs and the state of the game changes stochastically according to some transition matrix. The utility to a player is typically assumed to be the discounted sum of the payoffs per round. Stochastic games are known to admit Nash equilibria \cite{shapley1953stochastic}. A key point of difference between stochastic games and RUC games is that literature on stochastic games only considers stationary strategies, whereas strategies in our RUC games do not have this restriction.

\section{Preliminaries}
\label{sec:prelims}

\subparagraph{Notation.} We introduce some relevant notation.
\begin{enumerate}
\item For any boolean proposition $P$, let $\boolone(P)$ be 1 if $P$ is true and 0 otherwise.
\item For $n \in \mathbb{Z}_{\ge 0}$, let $\toN \defeq \{1, 2, \ldots, n\}$.
\item For any matrix $A \in \mathbb{R}^{m \times n}$, let $A_{i,j}$ or $A[i,j]$ denote the entry
    in the $i\Th$ row and $j\Th$ column.
    For any vector $\vvec \in \mathbb{R}^n$, let $\vvec_i$ denote the $i\Th$ entry of the vector.
\item For $n \in \mathbb{Z}_{\ge 1}$, let
    $\simplex_n \defeq \{\xvec \in \mathbb{R}_{\ge 0}^n: \sum_{i=1}^n \xvec_i = 1\}$.
\item For any vector $\vvec \in \mathbb{R}^n$, let $\support(\vvec) \defeq \{i \in \toN: \vvec_i \neq 0\}$.
    $\vvec$ is said to have \emph{full support} if $\support(\vvec) = \toN$.
\item For a vector $\vvec \in \mathbb{R}^n$, define $\|\vvec\|_1 \defeq \sum_{i=1}^n |\vvec_i|$.
\item For any $i \in \mathbb{Z}_{\ge 1}$, let $\evec^{(i)}$ be a vector such that $\evec^{(i)}_i = 1$
    and $\evec^{(i)}_j = 0$ for all $j \neq i$.
\end{enumerate}

\subparagraph{Two-Player Games.}
We study two-player games between a \textit{\batter{}}, who is interested in maximizing her payoff, and a \textit{\bowler{}}, who is interested in minimizing her payoff. We let $\Xcal$ and $\Ycal$ denote the strategy space of the \batterAndBowler{}, respectively. When the \batterAndBowler{} use strategies $x\in \Xcal$ and $y\in \Ycal$, respectively, the \batter{} gets a score $f_1(x,y)$ and the \bowler{} incurs a cost of $f_2(x,y)$, where these payoffs are given by functions $f_1, f_2: \Xcal \times \Ycal \to \mathbb{R} \cup \{\infty\}$. When $f_1 = f_2$, the game is said to be a \emph{zero-sum} game.

\subparagraph{Nash equilibrium.} A Nash equilibrium (NE) is pair of strategies where no player can improve her payoff by unilaterally changing her strategy.
\begin{definition}
Let $f_1, f_2: \Xcal \times \Ycal \to \mathbb{R} \cup \{\infty\}$ be the players' payoff functions.
The pair $(x^*, y^*) \in \Xcal \times \Ycal$ is called a \emph{Nash equilibrium} if
no player can improve their payoff by switching to a different strategy. Formally,
\begin{enumerate}
\item (Max player cannot improve) $\forall x \in \Xcal$, $f_1(x, y^*) \le f_1(x^*, y^*)$.
\item (Min player cannot improve) $\forall y \in \Ycal$, $f_2(x^*, y) \ge f_2(x^*, y^*)$.
\end{enumerate}
\end{definition}

While NE always exist for bimatrix games, they need not exist for general games. Moreover, NE may not be unique, i.e., multiple distinct NE may exist. We now define approximate NE.

\begin{definition}
Let $f_1, f_2: \Xcal \times \Ycal \to \mathbb{R} \cup \{\infty\}$ be the players' payoff functions.
The pair $(x^*, y^*) \in \Xcal \times \Ycal$ is called an \emph{$\eps$-approximate Nash equilibrium} if
no player can improve their payoff by a factor of $(1+\eps)$ by switching to a different strategy. Formally,
\begin{enumerate}
\item (Max player cannot improve) $\forall x \in \Xcal$, $f_1(x, y^*) \le (1+\eps)\cdot f_1(x^*, y^*)$.
\item (Min player cannot improve) $\forall y \in \Ycal$, $f_2(x^*, y) \ge (1+\eps)^{-1}\cdot f_2(x^*, y^*)$.
\end{enumerate}
\end{definition}

\subparagraph{Zero-Sum Games.}

For a two-player zero-sum game, the following theorem shows that all Nash equilibria give the same payoffs to the agents,
so the agents don't prefer any one of them over the other.

\begin{proposition}
\label{thm:all-ne-same-value}
Let $(f,f)$ be a zero-sum game, where $f: \Xcal \times \Ycal \to \mathbb{R} \cup \{\infty\}$.
If $(f,f)$ admits multiple Nash equilibria, they have the same payoff.
Formally, if $(x_1, y_1)$ and $(x_2, y_2)$ are Nash equilibria, then
$f(x_1, y_1) = f(x_2, y_2) = f(x_1, y_2) = f(x_2, y_1)$.
Moreover, $(x_1, y_2)$ and $(x_2, y_1)$ are also Nash equilibria.
\end{proposition}
\begin{proof}
Since $(x_1, y_1)$ is a Nash equilibrium, the \batter{} doesn't gain by switching to $x_2$
and \bowler{} doesn't gain by switching to $y_2$. Hence,
$f(x_2, y_1) \le f(x_1, y_1) \le f(x_1, y_2)$.
Since $(x_2, y_2)$ is a Nash equilibrium, the \batter{} doesn't gain by switching to $x_1$
and the \bowler{} doesn't gain by switching to $y_1$. Hence,
$f(x_1, y_2) \le f(x_2, y_2) \le f(x_2, y_1)$.
Combining these inequalities gives us $f(x_1, y_1) = f(x_2, y_2) = f(x_1, y_2) = f(x_2, y_1)$.

For any $x$ and $y$, we get $f(x, y_2) \le f(x_2, y_2) = f(x_1, y_2) = f(x_1, y_1) \le f(x_1, y)$,
and $f(x, y_1) \le f(x_1, y_1) = f(x_2, y_1) = f(x_2, y_2) \le f(x_2, y)$.
Hence, $(x_1, y_2)$ and $(x_2, y_1)$ are also Nash equilibria.
\end{proof}

\section{Nash Equilibria in Stationary RUC games}
\label{sec:meml}

In this section, we study Nash equilibria for SRUC games, i.e.,
an RUC game $(A,B)$ where agents are \textit{forced} to use only stationary strategies.
Let $e^{(A)}(\xvec, \yvec)$ and $e^{(B)}(\xvec, \yvec)$
be the \batter{}'s expected total score and the \bowler{}'s expected total cost,
when the \batterAndBowler{} play stationary strategies $\xvec$ and $\yvec$, respectively.

\begin{restatable}{lemma}{rthmExpectedScore}
\label{thm:expected-score}
Let $(A, B)$ represent an SRUC game.
Let $\xvec$ and $\yvec$ be the \batter{}'s and \bowler{}'s stationary strategies, respectively.
Then for $C \in \{A, B\}$, we have
\[ e^{(C)}(\xvec, \yvec) = \begin{cases}
\displaystyle \frac{\xvec\trn C\yvec}{\xvec\trn\yvec} & \textrm{ if } \xvec\trn\yvec > 0
\\ \infty & \textrm{ if } \xvec\trn\yvec = 0 \textrm{ and } \xvec\trn C\yvec > 0
\\ 0 & \textrm{ if } \xvec\trn\yvec = \xvec\trn C\yvec = 0
\end{cases}. \]
\end{restatable}
\begin{proof}[Proof sketch]
(See \cref{sec:exp-and-var} for the full proof.)

The \batter{}'s expected per-round score is $\xvec\trn A\yvec$,
and the probability of collision in a round is $\xvec\trn\yvec$.
If $\xvec\trn\yvec = 0$, a collision never happens, and so
her total score is $\infty$ or $0$ depending on whether $\xvec\trn A\yvec > 0$.
If $\xvec\trn\yvec > 0$, then we can find $\mu \defeq e^{(A)}(\xvec, \yvec)$
by solving the equation $\mu = \xvec\trn A\yvec + (1-\xvec\trn\yvec)\mu$,
which gives $\mu = \xvec\trn A\yvec/\xvec\trn\yvec$.
The \bowler{}'s expected total cost can be found analogously.
\end{proof}

From now on, instead of looking at SRUC games as multi-round games, we will treat them like
single round games where the strategy space is $\simplex_n$ for both players,
and the payoff functions are $e^{(A)}$ and $e^{(B)}$
for the \batter{} and \bowler{}, respectively.

\subsection{Existence of Nash Equilibrium}

This section shows the existence of a Nash equilibrium in SRUC games. We first recall the definitions of an eigenvalue and eigenvector of a matrix.

\begin{definition}
\label{thm:eigen}
For a square matrix $A \in \mathbb{R}^n$, complex number $\lambda$,
and complex vector $\vvec \in \mathbb{C}^n - \{\zerovec\}$,
$(\lambda, \vvec)$ is called an \emph{eigenpair} of $A$
(and $\lambda$ is called an \emph{eigenvalue} of $A$
and $\vvec$ is called an \emph{eigenvector} of $A$) if $A\vvec = \lambda v$.
\end{definition}

\begin{lemma}
For a square matrix $A$, $\lambda$ is an eigenvalue of $A$ iff $\lambda$ is an eigenvalue of $A\trnF$.
\end{lemma}
\begin{proof}
$\lambda$ is an eigenvalue of $A$ iff $\det(A - \lambda I) = 0$,
and $\det(A - \lambda I) = \det(A\trn - \lambda I)$.
\end{proof}

The following lemma establishes a sufficient condition for a strategy pair to be a Nash equilibrium of an SRUC game.

\begin{lemma}
\label{thm:pos-eig-gives-ne}
Consider the SRUC game $(A, B)$.
Let $(\alpha, \ystar)$ be an eigenpair for $A$ such that
$\|\ystar\|_1 = 1$ and $\ystar_i > 0$ for all $i \in \toN$.
Let $(\beta, \xstar)$ be an eigenpair for $B\trnF$ such that
$\|\xstar\|_1 = 1$ and $\xstar_i > 0$ for all $i \in \toN$.
Then $(\xstar, \ystar)$ is a Nash equilibrium.
Moreover, for any $\xvec, \yvec \in \simplex_n$, we have
$e^{(A)}(\xvec, \ystar) = \alpha$ and $e^{(B)}(\xstar, \yvec) = \beta$.
\end{lemma}
\begin{proof}
Let $\xvec, \yvec \in \simplex_n$. Since $\xstar$ and $\ystar$ have full support,
$\xstarT\yvec > 0$ and $\xvec\trn\ystar > 0$. Also, by \cref{thm:expected-score},
\begin{align*}
e^{(A)}(\xvec, \ystar) &= \frac{\xvec\trn A\ystar}{\xvec\trn\ystar}
    = \frac{\xvec\trn(\alpha\ystar)}{\xvec\trn\ystar} = \alpha.
& e^{(B)}(\xstar, \yvec) &= \frac{\xstarT B\yvec}{\xstarT\yvec}
    = \frac{(\beta\xstar)\trn\yvec}{\xstarT\yvec} = \beta.
\end{align*}
Hence, $(\xstar, \ystar)$ is a Nash equilibrium.
\end{proof}

We now show that such strategy pairs exist due to the Perron-Frobenius theorem stated below, thus proving the existence of Nash equilibrium for SRUC games in \cref{thm:ne-exists}.

\begin{theorem}[Perron-Frobenius~\cite{perronfrobenius}]
\label{thm:perron-frobenius}
Let $A \in \mathbb{R}_{\ge 0}^{n \times n}$ be an irreducible matrix. Then
\begin{enumerate}
\item There exists a (unique) positive eigenvalue $\rho$ of $A$, called the \emph{Perron root} of $A$,
    such that for any other (possibly complex) eigenvalue $\lambda$ of $A$, $|\lambda| \le \rho$.
\item There exist unique vectors $\uvec$ and $\vvec$ such that
    $A\trn\uvec = \rho\uvec$, $A\vvec = \rho\vvec$, and $\sum_{i=1}^n \uvec_i = \sum_{j=1}^n \vvec_j = 1$.
    $\uvec$ and $\vvec$ are called the left and right \emph{Perron vectors} of $A$, respectively.
\item $\uvec_i > 0$ and $\vvec_i > 0$ for all $i \in \toN$.
\item $\displaystyle
    \min_{i=1}^n \sum_{j=1}^n A_{i,j} \le \rho \le \max_{i=1}^n \sum_{j=1}^n A_{i,j},
    \quad \textrm{and} \quad
    \min_{j=1}^n \sum_{i=1}^n A_{i,j} \le \rho \le \max_{j=1}^n \sum_{i=1}^n A_{i,j}$.
\end{enumerate}
\end{theorem}

\begin{definition}
For irreducible matrices $A$ and $B$, let $\perronSolution(A, B)$ be the tuple
$(\rho_A, \rho_B, \xstar, \ystar)$, where
$\rho_A$ is the Perron root of $A$, $\ystar$ is the right Perron vector of $A$,
$\rho_B$ is the Perron root of $B$, and $\xstar$ is the left Perron vector of $B$.
\end{definition}

\begin{theorem}
\label{thm:ne-exists}
For the SRUC game $(A, B)$, where $A$ and $B$ are irreducible matrices,
let $(\rho_A, \rho_B, \xstar, \ystar) = \perronSolution(A, B)$.
Then $(\xstar, \ystar)$ is a Nash equilibrium.
Moreover, for any $\xvec, \yvec \in \simplex_n$, we have
$e^{(A)}(\xvec, \ystar) = \rho_A$ and $e^{(B)}(\xstar, \yvec) = \rho_B$.
\end{theorem}
\begin{proof}
Follows from \cref{thm:perron-frobenius,thm:pos-eig-gives-ne}.
\end{proof}

\subsection{Uniqueness of Nash Equilibrium}

In this section, we show that an SRUC game has a unique NE under mild assumptions on the payoff matrices. We begin with two lemmas which together show that any NE of an SRUC game has both players using strategies with full support, provided the payoff matrices $A$ and $B$ are irreducible and $\graph(A)$ is a subgraph of $\graph(B)$.

\begin{lemma}
\label{thm:ystar-full-supp}
Let $(\xstar, \ystar)$ be a Nash equilibrium for the SRUC game $(A, B)$,
where $A$ is irreducible and $\graph(A)$ is a subgraph of $\graph(B)$.
Then $\support(\ystar) = \toN$.
\end{lemma}
\begin{proof}
The key idea is that if $\ystar$ doesn't have full support, then the \batter{} is incentivized to
play a strategy outside $\support(\ystar)$ to get an infinite score.
Such a strategy would impose an infinite cost on the \bowler{}.
Finally, it's possible for the \bowler{} to have a bounded cost by playing a full-support strategy.

For any $C \in \{A, B\}$, let $(\toN, E_C) \defeq \graph(C)$.
Let $\Ebar_C \defeq \toN \times \toN \setminus E_C$.
For any $\xvec, \yvec \in \simplex_n$ such that $\xvec\trn\yvec = 0$, we get
$\xvec\trn C\yvec = 0 \iff \support(\xvec) \times \support(\yvec) \subseteq \Ebar_C$.
Since $E_A \subseteq E_B$, we get $\xvec\trn A\yvec > 0 \implies \xvec\trn B\yvec > 0$.

Let $S_y \defeq \support(\ystar)$. Suppose $S_y \neq \toN$.
Then $\exists (i, j) \in E_A$ from $\toN \setminus S_y$ to $S_y$,
since $A$ is irreducible. Hence,
$(\evec^{(i)})\trn A\ystar = \sum_{k \in S_y}A_{i,k}\ystar_k \ge A_{i,j}\ystar_j > 0$.
Also, $(\evec^{(i)})\trn\ystar = 0$, since $i \not\in S_y$.
Hence, $e^{(A)}(\evec^{(i)}, \ystar) = \infty$.

Since $(\xstar, \ystar)$ is a Nash equilibrium,
$\infty = e^{(A)}(\evec^{(i)}, \ystar) \le e^{(A)}(\xstar, \ystar)$.
This is only possible if $\xstarT\ystar = 0$ and $\xstarT A\ystar > 0$.
This means $\xstarT B\ystar > 0$, and so $e^{(B)}(\xstar, \ystar) = \infty$.

Let $\yhat_j = 1/n$ for all $j \in \toN$. Then $\xstarT\yhat > 0$,
so $e^{(B)}(\xstar, \yhat)$ is finite.
Since $(\xstar, \ystar)$ is a Nash equilibrium, we get
$\infty = e^{(B)}(\xstar, \ystar) \le e^{(B)}(\xstar, \yhat) \neq \infty$.
This is a contradiction, so $\support(\ystar) = \toN$.
\end{proof}

\begin{lemma}
\label{thm:xstar-full-supp}
Let $(\xstar, \ystar)$ be a Nash equilibrium for SRUC game $(A, B)$,
where $B$ is irreducible and $\support(\ystar) = \toN$. Then $\support(\xstar) = \toN$.
\end{lemma}
\begin{proof}
The key idea is that if $\xstar$ doesn't have full support, then the \bowler{}
can reduce her cost by not playing actions outside $\support(\xstar)$.

Let $S_x \defeq \support(\xstar)$. Let $\alpha \defeq \sum_{j \in S_x} \ystar_j$.
Then $1-\alpha = \sum_{j \in \toN \setminus S_x} \ystar_j$.
Suppose $S_x \neq \toN$. Then $0 < \alpha < 1$, since $\support(\ystar) = \toN$.
Define vectors $\yhat$ and $\ytild$ as follows:
\begin{align*}
\yhat_j &= \begin{cases}\ystar_j/\alpha & \textrm{ if } j \in S_x
\\ 0 & \textrm{ if } j \not\in S_x\end{cases},
& \ytild_j &= \begin{cases}0 & \textrm{ if } j \in S_x
\\ \ystar_j/(1-\alpha) & \textrm{ if } j \not\in S_x\end{cases}.
\end{align*}
Then $\ystar = \alpha\yhat + (1-\alpha)\ytild$ and $\xstarT\ystar = \alpha\xstarT\yhat$. Hence,
\[ e^{(B)}(\xstar, \ystar) = \frac{\xstarT B\ystar}{\xstarT\ystar}
= \frac{\alpha\xstarT B\yhat + (1-\alpha)\xstarT B\ytild}{\alpha\xstarT\yhat}
= e^{(B)}(\xstar, \yhat) + \frac{1-\alpha}{\alpha}\,\frac{\xstarT B\ytild}{\xstarT\yhat}. \]
Since $(\xstar, \ystar)$ is a Nash equilibrium,
$e^{(B)}(\xstar, \ystar) \le e^{(B)}(\xstar, \yhat)$. Hence, $\xstarT B\ytild = 0$.
However,
\[ \xstarT B\ytild = \sum_{i \in S_x}\sum_{j \in \toN\setminus S_x} \xstar_iB_{i,j}\ytild_j. \]
Since $B$ is irreducible, there is an edge $(i, j)$ in $\graph(B)$
from $S_x$ to $\toN \setminus S_x$.
Therefore, $\xstarT B\ytild \ge B_{i,j}\xstar_i\ytild_j > 0$,
which is a contradiction. Hence, $\support(\xstar) = \toN$.
\end{proof}

We now state a simple result about the ratio of sums (\cref{thm:sum-ratio-min-max}),
and use it to prove that full support Nash equilibria must be eigenvectors (\cref{thm:meml-ne-is-eigen}).

\begin{lemma}
\label{thm:sum-ratio-min-max}
Let $\uvec \in \mathbb{R}^n$ and $\vvec \in \mathbb{R}_{>0}^n$. Let
\begin{align*}
\alpha &\defeq \min_{i=1}^n \frac{\uvec_i}{\vvec_i}
& \beta &\defeq \max_{i=1}^n \frac{\uvec_i}{\vvec_i}
& z &\defeq \frac{\sum_{i=1}^n \uvec_i}{\sum_{i=1}^n \vvec_i}.
\end{align*}
Then either $z = \alpha = \beta$ or $\alpha < z < \beta$.
\end{lemma}
\newcommand*{\longNegSpace}{\!\!\!\!\!\!}
\begin{proof}
If $\alpha = \beta$, then $\uvec_i/\vvec_i = \alpha$ for all $i \in \toN$, and hence $z = \alpha$.
Now let $\alpha < \beta$. Then $\alpha\vvec_i \le \uvec_i \le \beta\vvec_i$ for all $i \in \toN$.
Pick $p$ and $q$ such that $\uvec_p/\vvec_p = \alpha$ and $\uvec_q/\vvec_q = \beta$. Then
\begin{align*}
&\sum_{i=1}^n \uvec_i \le \alpha \vvec_p
    + \longNegSpace\sum_{i \in \toN \setminus \{p\}}\longNegSpace \beta\vvec_i
    < \beta\sum_{i=1}^n \vvec_i,
& &\sum_{i=1}^n \uvec_i \ge \longNegSpace\sum_{i \in \toN \setminus \{q\}}\longNegSpace \alpha\vvec_i + \beta\vvec_q
    > \alpha\sum_{i=1}^n \vvec_i.
\end{align*}
Hence, $\alpha < z < \beta$.
\end{proof}

\begin{lemma}
\label{thm:meml-ne-is-eigen}
Let $(\xstar, \ystar)$ be a Nash equilibrium for SRUC game $(A, B)$
such that $\support(\xstar) = \support(\ystar) = \toN$.
Then $(e^{(B)}(\xstar, \ystar), \xstar)$ is an eigenpair of $B\trnF$
and $(e^{(A)}(\xstar, \ystar), \ystar)$ is an eigenpair of $A$.
\end{lemma}
\begin{proof}
For $i \in \toN$, let $\wvec_i \defeq \xstar_i\ystar_i$.
Then $\wvec_i > 0$ for all $i \in \toN$.

For $j \in \toN$, let $\uvec_j \defeq (B\trn\xstar)_j\ystar_j$. Let
\begin{align*}
\alpha_x &\defeq \min_{j=1}^n \frac{\uvec_j}{\wvec_j} = \min_{j=1}^n \frac{(B\trn\xstar)_j}{\xstar_j}
& \beta_x &\defeq \max_{j=1}^n \frac{\uvec_j}{\wvec_j} = \max_{j=1}^n \frac{(B\trn\xstar)_j}{\xstar_j}
\end{align*}
Let $p$ and $q$ be indices such that $\uvec_p/\wvec_p = \alpha_x$ and $\uvec_q/\wvec_q = \beta_x$.
Then $e^{(B)}(\xstar, \evec^{(p)}) = \alpha_x$, $e^{(B)}(\xstar, \evec^{(q)}) = \beta_x$,
and $e^{(B)}(\xstar, \ystar) = (\sum_{j=1}^n \uvec_j)/(\sum_{j=1}^n \wvec_j)$.
Suppose $\alpha_x \neq \beta_x$. Then by \cref{thm:sum-ratio-min-max},
$e^{(B)}(\xstar, \evec^{(p)}) < e^{(B)}(\xstar, \ystar) < e^{(B)}(\xstar, \evec^{(q)})$.
This contradicts the fact that $(\xstar, \ystar)$ is a Nash equilibrium.
Hence, $\alpha_x = \beta_x = e^{(B)}(\xstar, \ystar)$
and $(B\trn\xstar)_j/\xstar_j = \alpha_x$ for all $j \in \toN$.
Hence, $(\alpha_x, \xstar)$ is an eigenpair of $B\trnF$.

For $i \in \toN$, let $\vvec_i \defeq \xstar_i(A\ystar)_i$. Let
\begin{align*}
\alpha_y &\defeq \min_{i=1}^n \frac{\vvec_i}{\wvec_i} = \min_{i=1}^n \frac{(A\ystar)_i}{\ystar_i}
& \beta_y &\defeq \max_{i=1}^n \frac{\vvec_i}{\wvec_i} = \max_{i=1}^n \frac{(A\ystar)_i}{\ystar_i}
\end{align*}
We can similarly show that $\alpha_y = \beta_y = e^{(A)}(\xstar, \ystar)$
and $(\alpha_y, \ystar)$ is an eigenpair of $A$.
\end{proof}

Finally, we show that non-negative eigenvectors are essentially Perron vectors (\cref{thm:nonneg-eig-is-perron}),
and use this result to establish uniqueness of Nash equilibrium (\cref{thm:meml-ne-unique}).

\begin{lemma}
\label{thm:nonneg-eig-is-perron}
Let $A \in \mathbb{R}_{\ge 0}^{n \times n}$ be an irreducible matrix.
Let $\rho$ be the Perron root of $A$, and $\uvec$ and $\vvec$ be the
left and right Perron vectors of $A$, respectively.
Then both of these hold:
\begin{enumerate}
\item \label{item:nonneg-eig-is-perron:A}If
    $(\lambda, \yvec)$ is an eigenpair of $A$ such that $\yvec_i \ge 0$ $\forall i$
    and $\|\yvec\|_1 = 1$, then $\lambda = \rho$ and $\yvec = \vvec$.
\item \label{item:nonneg-eig-is-perron:AT}If
    $(\lambda, \xvec)$ is an eigenpair of $A\trnF$ such that $\xvec_i \ge 0$ $\forall i$
    and $\|\xvec\|_1 = 1$, then $\lambda = \rho$ and $\xvec = \uvec$.
\end{enumerate}
\end{lemma}
\begin{proof}
$\rho\uvec\trn\yvec = (A\trn\uvec)\trn\yvec = \uvec\trn(A\yvec) = \lambda\uvec\trn\yvec$.
Since $\support(\uvec) = \toN$, we get $\uvec\trn\yvec > 0$. Hence, $\lambda = \rho$.
By \cref{thm:perron-frobenius}, $\vvec$ is the unique eigenvector corresponding to
eigenvalue $\rho$ such that $\|\vvec\|_1 = 1$. Hence, $\yvec = \vvec$.
We get part \ref{item:nonneg-eig-is-perron:AT} of the lemma by applying
part \ref{item:nonneg-eig-is-perron:A} on $A\trnF$.
\end{proof}

\begin{theorem}
\label{thm:meml-ne-unique}
Let $(A, B)$ be an SRUC game, where $A$ and $B$ are irreducible and
$\graph(A)$ is a subgraph of $\graph(B)$.
Let $(\rho_A, \rho_B, \xstar, \ystar) = \perronSolution(A, B)$.
Then $(\xstar, \ystar)$ is the unique Nash equilibrium for this SRUC game.
\end{theorem}
\begin{proof}
Let $(\xhat, \yhat)$ be any Nash equilibrium for SRUC game $(A, B)$.
From \cref{thm:ystar-full-supp,thm:xstar-full-supp}, we get
$\support(\xhat) = \support(\yhat) = \toN$.

Let $\sigma_C \defeq e^{(C)}(\xhat, \yhat)$ for $C \in \{A, B\}$.
By \cref{thm:meml-ne-is-eigen},
$(\sigma_B, \xhat)$ is an eigenpair of $B\trnF$ and
$(\sigma_A, \yhat)$ is an eigenpair of $A$.
By \cref{thm:nonneg-eig-is-perron}, we get $\sigma_B = \rho_B$, $\xhat = \xstar$,
$\sigma_A = \rho_A$, $\yhat = \ystar$.
\end{proof}

\Cref{thm:meml-ne-unique} shows uniqueness of NE of an SRUC game $(A,B)$ under two conditions:
(i) $A$ and $B$ are irreducible, and (ii) $\graph(A) \subseteq \graph(B)$.
The following two lemmas show that if either of these conditions is relaxed, we can no longer guarantee uniqueness of NE.

\begin{lemma}
\label{thm:edge-in-A-not-in-B}
Let $(A, B)$ be an SRUC game. Suppose $\exists (i, j)$ such that $i \neq j$,
$A_{i,j} > 0$, and $B_{i,j} = 0$. Then $(\evec^{(i)}, \evec^{(j)})$ is a Nash equilibrium.
Moreover, $e^{(A)}(\evec^{(i)}, \evec^{(j)}) = \infty$ and $e^{(B)}(\evec^{(i)}, \evec^{(j)}) = 0$.
\end{lemma}
\begin{proof}
${\evec^{(i)}}\trn A\evec^{(j)} = A[i,j] > 0$, ${\evec^{(i)}}\trn B\evec^{(j)} = B[i, j] = 0$,
and ${\evec^{(i)}}\trn\evec^{(j)} = 0$.
\end{proof}

The NE in \cref{thm:edge-in-A-not-in-B} is deterministic
but the NE in \cref{thm:ne-exists} is randomized, implying that the game has at least two NE.
Hence, the condition $\graph(A) \subseteq \graph(B)$ is necessary for uniqueness
of NE if $A$ and $B$ are irreducible. The next lemma shows that if either $A$ or $B$ is reducible,
then uniqueness of NE is not guaranteed.

\begin{lemma}
Let $(A, A)$ be an SRUC game where $A \defeq \big(\begin{smallmatrix}1 & 0 \\ 0 & 2\end{smallmatrix}\big)$.
Then $(\xvec, \yvec)$ is a Nash equilibrium iff $\support(\xvec) = \{1, 2\}$ and $\yvec = [1, 0]$.
\end{lemma}
\begin{proof}
Observe that $A$ is reducible. For any $\xvec, \yvec \in \simplex_2$, if $\xvec_2 > 0$ and $\yvec_2 > 0$, then
$e^{(A)}(\xvec, \yvec) = 1 + \xvec_2\yvec_2/(\xvec_1\yvec_1 + \xvec_2\yvec_2) > 1$.
When $\xvec\trn\yvec = 0$, we get $e^{(A)}(\xvec, \yvec) = 0$. Otherwise, $e^{(A)}(\xvec, \yvec) = 1$.
This is summarized by the following table:

\begin{tabular}{c|ccc}
$x$ \textbackslash{} $y$ & $[1, 0]$ & mix & $[0, 1]$
\\\hline $[1,0]$ & $1$ & $1$ & $0$
\\ mix & $1$ & $>1$ & $>1$
\\ $[0,1]$ & $0$ & $>1$ & $>1$
\end{tabular}

This table shows that $(\xvec, \yvec)$ is a NE
iff $\support(\xvec) = \{1, 2\}$ and $\yvec = [1, 0]$.
\end{proof}

The above example also shows that unlike the guarantee of \cref{thm:ne-exists},
Nash equilibrium strategies may not have full support.

\subsection{Approximate Nash Equilibrium}

It may not be possible to compute Nash equilibria for SRUC games exactly,
e.g., if eigenvectors are irrational.
Hence, we would like to compute approximate Nash equilibria. We will use the following lemma.

\begin{lemma}
\label{thm:approx-eigen-norm}
Let $\xstar, \xhat \in \simplex_n$ and $\delta \in [0, 1)$ such that $\support(\xstar) = \toN$
and $|\xhat_i - \xstar_i| \le \delta\xstar_i$ for all $i \in \toN$. Then $\support(\xhat) = \toN$
and for any matrix $C \in \mathbb{R}_{\ge 0}^{n \times n}$ and any $j \in \toN$,
\[ \frac{(C\trn\xhat)_j}{\xhat_j} \in \left[\frac{1-\delta}{1+\delta}\frac{(C\trn\xstar)_j}{\xstar_j},
    \frac{1+\delta}{1-\delta}\frac{(C\trn\xstar)_j}{\xstar_j}\right]. \]
\end{lemma}
\begin{proof}
For any $i \in \toN$, we have $\xhat_i \in [(1-\delta)\xstar_i, (1+\delta)\xstar_i]$.
Hence, $\support(\xhat) = \toN$.
\begin{align*}
& |(C\trn\xhat)_j - (C\trn\xstar)_j| \le |(C\trn(\xhat-\xstar))_j|
\le \left|\sum_{i=1}^n C[i,j](\xhat_i-\xstar_i)\right|
\\ &\le \sum_{i=1}^n C[i,j]|\xhat_i-\xstar_i|
\le \delta\sum_{i=1}^n C[i,j]\xstar_i = \delta(C\trn\xstar)_j.
\\ &\implies (C\trn\xhat)_j \in [(1-\delta)(C\trn\xstar)_j, (1+\delta)(C\trn\xstar)_j]
\\ &\implies \frac{(C\trn\xhat)_j}{\xhat_j}
    \in \left[\frac{1-\delta}{1+\delta}\frac{(C\trn\xstar)_j}{\xstar_j},
        \frac{1+\delta}{1-\delta}\frac{(C\trn\xstar)_j}{\xstar_j}\right].
\qedhere
\end{align*}
\end{proof}

We now show that a close-enough approximation to a NE is an approximate NE of an SRUC game. This result allows us to compute approximate NE by using methods for approximately estimating the leading eigenvector of a matrix, such as power iteration \cite{poweriteration}.

\begin{theorem}
\label{thm:approx-ne}
Let $(\xstar, \ystar)$ be a Nash equilibrium for the SRUC game $(A, B)$
such that $\support(\xstar) = \support(\ystar) = \toN$.
For $\delta\in[0,1)$, let $\xhat, \yhat \in \simplex_n$ be such that for all $i \in \toN$,
$|\xhat_i - \xstar_i| \le \delta\xstar_i$ and $|\yhat_i - \ystar_i| \le \delta\ystar_i$.
Then $(\xhat, \yhat)$ is a $\frac{4\delta}{(1-\delta)^2}$-approximate Nash equilibrium.
\end{theorem}
\begin{proof}
By \cref{thm:meml-ne-is-eigen}, $(\rho_B, \xstar)$ is an eigenpair of $B\trnF$.
Hence, by \cref{thm:approx-eigen-norm}, for all $j \in \toN$,
\[ \frac{(B\trn\xhat)_j}{\xhat_j} \in \left[\frac{1-\delta}{1+\delta}\rho_B,
    \frac{1+\delta}{1-\delta}\rho_B\right]. \]
For any $\yvec \in \simplex_n$, using \cref{thm:sum-ratio-min-max}, we get
\[ e^{(B)}(\xhat, \yvec) = \frac{\sum_{j=1}^n (B\trn\xhat)_j\yvec_j}{\sum_{j=1}^n \xhat_j\yvec_j}
    \in \bigg[ \min_{j=1}^n \frac{(B\trn\xhat)_j}{\xhat_j}, \max_{j=1}^n \frac{(B\trn\xhat)_j}{\xhat_j}\bigg]
    = \bigg[\frac{1-\delta}{1+\delta}\rho_B, \frac{1+\delta}{1-\delta}\rho_B\bigg]. \]
In particular, for $\yvec = \yhat$, we get $e^{(B)}(\xhat, \yhat) \le \frac{1+\delta}{1-\delta}\rho_B$.
Let $\eps \defeq \frac{4\delta}{(1-\delta)^2}$. Thus, for any $\yvec \in \simplex_n$,
\[ e^{(B)}(\xhat, \yvec) \ge \frac{1-\delta}{1+\delta}\rho_B
    \ge \bigg(\frac{1-\delta}{1+\delta}\bigg)^2 e^{(B)}(\xhat, \yhat)
    = \frac{e^{(B)}(\xhat,\yhat)}{1+\eps}. \]
We can analogously show that for any $\xvec\in\simplex_n$,
\[ e^{(A)}(\xvec,\yhat) \le \bigg(\frac{1+\delta}{1-\delta}\bigg)^2 e^{(A)}(\xhat,\yhat)
    = (1+\eps)\cdot e^{(A)}(\xhat,\yhat). \]
Hence, $(\xhat,\yhat)$ is an $\eps$-approximate NE.
\end{proof}

\section{Nash Equilibria in General RUC games}
\label{sec:non-meml}

In this section, we study RUC games when players are allowed to use non-stationary strategies.
As discussed in the Introduction, the strategy space can be complicated in general RUC games.
We first describe a framework for defining strategies in a RUC game.

\subsection{Strategy Space of RUC games}

A deterministic (possibly non-stationary) strategy in an RUC game
is a function $f$ that takes as input a list of actions played
by the opponent so far and outputs the next action for the player.
E.g., if the \batter{} is using a deterministic strategy $f$,
then in the $k\Th$ round, the \batter{} will play the action $f(J)$,
where $J \defeq [j_1, j_2, \ldots, j_{k-1}]$ is the list of actions
played by the \bowler{} in the first $k-1$ rounds.

A randomized strategy is a distribution over deterministic strategies.
Formally, let $\toN^*$ be the set of all finite lists where each element is in $\toN$.
Let $\Omega$ be the set of all deterministic strategies, i.e.,
the set of all functions from $\toN^*$ to $\toN$.
Then a randomized strategy is given by a probability space $(\Omega, \Ecal, P)$.
Recall that in a probability space,
$\Ecal \subseteq 2^{\Omega}$ is the set of events
and $P: \Ecal \to [0,1]$ is a probability measure.
If a randomized strategy $f$ is sampled from this probability space, then for any set
$F \in \Ecal$ of deterministic strategies, we say $\Pr(f \in F) \defeq P(F)$.

For any $\xvec \in \simplex_n$, and any random variable $X \in \toN$,
we write $X \sampledFrom \xvec$ to say that $\Pr(X = i) = \xvec_i$ for all $i \in \toN$.
The length of a list $L$, denoted by $|L|$, is the number of elements in $L$.
Unless specified otherwise, assume all lists are finite.

\begin{definition}[stationary strategy]
\label{defn:stdmeml}
Let $\xvec \in \simplex_n$. Let $[I_0, I_1, \ldots]$ be an infinite sequence of
independent random variables where for each $t \in \mathbb{Z}_{\ge 0}$,
$I_t \in \toN$ and $I_t \sampledFrom \xvec$.
Then $\stdmeml(\xvec)$ is a strategy $f$ where $f(J) \defeq I_{|J|}$.
$\stdmeml(\xvec)$ is called the \emph{stationary strategy} for parameter $\xvec$.
\end{definition}

With a little abuse of notation, we will sometimes write $\xvec$ instead of $\stdmeml(\xvec)$.

\subparagraph{Residual Strategies.}

For any two lists $L_1$ and $L_2$, let $L_1 + L_2$ denote their concatenation.
Let $\emptyset$ denote the empty list.
Let $[x]$ denote a list of length 1 containing the element $x$.

Given a list $K \in \toN^*$ and a strategy $f$, define the function $f^K$ as
$f^K(K') \defeq f(K + K')$. Intuitively, if a player is using strategy $f$,
then after the opponent has played actions $K$, $f^K$ is the strategy for the remaining game.
$f^K$ is called the \emph{$K$-residual} strategy of $f$.
Due to the recursive nature of RUC games, residual strategies are helpful in their analysis.

Let $I \defeq [i_1, \ldots, i_k]$ and $J \defeq [j_1, \ldots, j_k]$ be lists.
Let $f$ be a randomized strategy. We want to define $f^{(I,J)}$ as the strategy $f^J$
conditioned on the player responding with actions $I$ when the opponent plays actions $J$.
Formally, let $\isResponse(f, I, J)$ be the event that a player using strategy $f$ responds
with actions $I$ to opponent's actions $J$, i.e.,
$\forall t \in \toN[k], f([j_1, \ldots, j_{t-1}]) = i_t$.
Call $(I, J)$ a \emph{feasible history} for $f$ if $\Pr(\isResponse(f, I, J)) > 0$.
For any feasible history $(I, J)$ of $f$, define $f^{(I,J)}$,
called the \emph{$(I, J)$-residual strategy} of $f$, as a strategy having distribution
\[ \Pr(f^{(I,J)} \in F) \defeq \Pr(f^J \in F \mid \isResponse(f, I, J)). \]

Intuitively, stationary strategies should remain unchanged when conditioned on past actions.
We prove this formally.

\begin{restatable}{lemma}{rthmCondOnMeml}
\label{thm:cond-on-meml}
$\stdmeml(\xvec)^{(I,J)}$ has the same distribution as $\stdmeml(\xvec)$.
\end{restatable}
\movedProof{\cref{sec:non-meml-extra:str}}

\subparagraph{Expected Score.}
Consider the RUC game $(A, B)$, where $A$ and $B$ are $n \times n$ matrices.
Assume that the \batter{}'s and \bowler{}'s strategies are independent.
We now formally define the \emph{expected score (or cost) of a pair of strategies},
and obtain a recursive expression for it to make analysis easier.

\begin{definition}[score]
Let $S^{(A, r)}(f, g)$ and $S^{(B, r)}(f, g)$ be the \batter{}'s total score
and \bowler{}'s total cost, respectively, in the first $r$ rounds of the RUC game $(A, B)$
when the \batter{} uses strategy $f$ and the \bowler{} uses strategy $g$.
(Note that when $f$ and $g$ are randomized strategies,
$S^{(A, r)}(f, g)$ and $S^{(B, r)}(f, g)$ are random variables.)
For $C \in \{A, B\}$, let $e^{(C, r)}(f, g) \defeq \E(S^{(C, r)}(f, g))$
and $e^{(C,\infty)}(f, g) \defeq \lim_{r \to \infty} e^{(C,r)}(f, g)$.
\end{definition}

\begin{observation}
The sequence $[S^{(C,r)}(f, g)]_{r=0}^{\infty}$ is monotonically increasing,
so the sequence $[e^{(C,r)}(f, g)]_{r=0}^{\infty}$ is also monotonically increasing.
By the monotone convergence theorem,
$e^{(C,\infty)}(f, g) \in \mathbb{R}_{\ge 0} \cup \{\infty\}$,
i.e., the sequence $[e^{(C,r)}(f, g)]_{r=0}^{\infty}$ either
has a non-negative limit or is unbounded.
\end{observation}

\begin{observation}
\label{thm:score-nmem}
Let $f$ and $g$ be strategies of the \batterAndBowler{}, respectively.
Let $u \defeq f(\emptyset)$ and $v \defeq g(\emptyset)$. Then for $C \in \{A, B\}$,
\[ S^{(C,r)}(f, g) = \begin{cases}
C[u, v] + \boolone(u \neq v)S^{(C,r-1)}(f^{[v]}, g^{[u]}) & \textrm{ if } r > 0
\\ 0 & \textrm{ otherwise}\end{cases}. \]
\end{observation}

\begin{restatable}{lemma}{rthmExpScoreNmem}
\label{thm:expected-score-nmem}
Let $f$ and $g$ be independent strategies of the \batterAndBowler{}, respectively.
Let $f(\emptyset) \sampledFrom \xvec$, $g(\emptyset) \sampledFrom \yvec$,
$S_x \defeq \support(\xvec)$, and $S_y \defeq \support(\yvec)$. Then for $C \in \{A, B\}$,
\[ e^{(C,r)}(f, g) = \begin{cases}
\displaystyle \xvec\trn C\yvec + \sum_{i \in S_x}\sum_{j \in S_y}\xvec_i\yvec_j
    \boolone(i \neq j)e^{(C,r-1)}(f^{([i],[j])}, g^{([j],[i])}) & \textrm{ if } r > 0
\\ 0 & \textrm{ otherwise} \end{cases}, \]
\[ e^{(C,\infty)}(f, g) =
\xvec\trn C\yvec + \sum_{i \in S_x}\sum_{j \in S_y}\xvec_i\yvec_j
    \boolone(i \neq j)e^{(C,\infty)}(f^{([i],[j])}, g^{([j],[i])}). \]
\end{restatable}
\begin{proof}[Proof sketch]
Follows from \cref{thm:score-nmem}.
See \cref{sec:non-meml-extra:str} for the full proof.
\end{proof}

\subsection{Existence of Nash Equilibrium}

To see whether a pair of stationary strategies can give us a Nash equilibrium,
we first investigate (in the next two lemmas) the upper and lower bounds on a player's payoff
when she is free to play any strategy (even non-stationary ones)
and her opponent uses a stationary strategy with full support.

\begin{lemma}
\label{thm:batter-nmem-ne}
Let $(A, B)$ be an RUC game. Let $\xhat \in \simplex_n$ such that $\support(\xhat) = \toN$.
Let $\alpha \defeq \min_{i=1}^n (B\trn\xhat)_i/\xhat_i$
and $\beta \defeq \max_{i=1}^n (B\trn\xhat)_i/\xhat_i$.
Then for any strategy $g$, we have $\alpha \le e^{(B,\infty)}(\xhat, g) \le \beta$.
\end{lemma}
\begin{proof}
By the definition of $\alpha$ and $\beta$, and by \cref{thm:sum-ratio-min-max},
we get that for any $\yvec \in \simplex_n$, we have
$\alpha\xhat\trn\yvec \le \xhat\trn B\yvec \le \beta\xhat\trn\yvec$.

Define the predicate $P(r): \forall g, e^{(B,r)}(\xhat, g) \le \beta$.
We will prove $P(r)$ by induction, and that would imply $P(\infty)$.
The base case holds, since $e^{(B,0)}(\xhat, g) = 0 \le \beta$.
Now fix $r \ge 1$, $g$, and let $g(\emptyset) \sampledFrom \yvec$
and $S_y \defeq \support(\yvec)$. Then by \cref{thm:expected-score-nmem},
\begin{align*}
e^{(B,r)}(\xhat, g) &= \xhat\trn B\yvec + \sum_{i=1}^n\sum_{j \in S_y} \xhat_i\yvec_j
    \boolone(i \neq j)e^{(B,r-1)}(\xhat, g^{([j],[i])})
\\ &\le \beta\xhat\trn\yvec + \sum_{i=1}^n\sum_{j=1}^n \xhat_i\yvec_j\boolone(i \neq j)\beta
    \;= \beta.
    \tag{by inductive hypothesis}
\end{align*}
This proves $P(r)$, and hence, also proves $P(\infty)$.

Next, we need to prove that $e^{(B,\infty)}(\xhat, g) \ge \alpha$ for all $g$.
Let $\mu \defeq \inf_g e^{(B,\infty)}(\xhat, g)$. Pick any $\eps > 0$.
Let $\ghat$ be a strategy such that $e^{(B,\infty)}(\xhat, \ghat) \le \mu + \eps$.
Let $\ghat(\emptyset) \sampledFrom \yhat$ and $S_y \defeq \support(\yhat)$.
Let $\gamma \defeq \min_{i=1}^n \xhat_i$.
Then $\xhat\trn\yhat \ge \gamma$, and by \cref{thm:expected-score-nmem},
\begin{align*}
\mu + \eps &\ge e^{(B,\infty)}(\xhat, \ghat)
= \xhat\trn B\yhat + \sum_{i=1}^n\sum_{j \in S_y}\xhat_i\yhat_j
    \boolone(i \neq j)e^{(B,\infty)}(\xhat, \ghat^{([j],[i])})
\\ &\ge \alpha\xhat\trn\yhat + \sum_{i=1}^n\sum_{j=1}^n \xhat_i\yhat_j\boolone(i \neq j)\mu
\;= \mu + \xhat\trn\yhat(\alpha - \mu).
\end{align*}
Since $\support(\xhat) = \toN$, we get $\gamma > 0$ and
$\alpha - \mu \le \eps/\xhat\trn\yhat \le \eps/\gamma$.
Since this is true for all $\eps > 0$, we get $\alpha \le \mu$.
Hence, $\alpha \le e^{(B,\infty)}(\xhat, g) \le \beta$ for every strategy $g$.
\end{proof}

\begin{lemma}
\label{thm:bowler-nmem-ne}
Let $(A, B)$ be an RUC game. Let $\yhat \in \simplex_n$ such that $\support(\yhat) = \toN$.
Let $\alpha \defeq \min_{i=1}^n (A\yhat)_i/\yhat_i$
and $\beta \defeq \max_{i=1}^n (A\yhat)_i/\yhat_i$.
Then for any strategy $f$, we have $\alpha \le e^{(A,\infty)}(f, \yhat) \le \beta$.
\end{lemma}
\begin{proof}
(Similar to the proof of \cref{thm:batter-nmem-ne}.)
\end{proof}

Next, we identify sufficient conditions (as in \cref{thm:pos-eig-gives-ne} in \cref{sec:meml})
to make the upper and lower bounds in \cref{thm:batter-nmem-ne,thm:bowler-nmem-ne} coincide,
which gives us a Nash equilibrium.

\begin{lemma}
\label{thm:pos-eig-gives-ne-2}
Consider the RUC game $(A, B)$.
Let $(\alpha, \ystar)$ be an eigenpair for $A$ such that
$\|\ystar\|_1 = 1$ and $\ystar_i > 0$ for all $i \in \toN$.
Let $(\beta, \xstar)$ be an eigenpair for $B\trnF$ such that
$\|\xstar\|_1 = 1$ and $\xstar_i > 0$ for all $i \in \toN$.
Then $(\xstar, \ystar)$ is a Nash equilibrium.
Moreover, for any strategies $f$ and $g$, we have
$e^{(A,\infty)}(f, \ystar) = \alpha$ and $e^{(B,\infty)}(\xstar, g) = \beta$.
\end{lemma}
\begin{proof}
Follows from \cref{thm:batter-nmem-ne,thm:bowler-nmem-ne}.
\end{proof}

Next, we show that conditions of \cref{thm:pos-eig-gives-ne-2}
can be satisfied using the Perron-Frobenius theorem (\cref{thm:perron-frobenius}),
so a Nash equilibrium given by stationary strategies always exists.

\begin{theorem}
\label{thm:ne-exists-2}
Let $(A, B)$ be an RUC game where $A$ and $B$ are irreducible (c.f.~\cref{defn:irred}).
Let $(\rho_A, \rho_B, \xstar, \ystar) = \perronSolution(A, B)$.
Then $(\xstar, \ystar)$ is a Nash equilibrium.
Moreover, for any strategies $f$ and $g$, we have
$e^{(A,\infty)}(f, \ystar) = \rho_A$ and $e^{(B,\infty)}(\xstar, g) = \rho_B$.
\end{theorem}
\begin{proof}
Follows from \cref{thm:pos-eig-gives-ne-2}.
\end{proof}

Note that for any stationary strategies $\xvec$ and $\yvec$,
if $(\xvec, \yvec)$ is a Nash equilibrium for the RUC game $(A, B)$,
then $(\xvec, \yvec)$ is also a Nash equilibrium for the SRUC game $(A, B)$.
This is because if a player cannot improve her payoff by switching to a different strategy,
then she also cannot improve her payoff by switching to a different stationary strategy.
Hence, \cref{thm:ne-exists} is a corollary of \cref{thm:ne-exists-2}.

Since computing Nash equilibrium exactly may be hard,
we consider approximate Nash equilibria. To do this,
we generalize the corresponding result about SRUC games
(\cref{thm:approx-ne} in \cref{sec:meml}) to RUC games.

\begin{restatable}[approximate Nash equilibrium]{theorem}{rthmApproxNeNmem}
\label{thm:approx-ne-nmem}
Let $(\stdmeml(\xstar), \stdmeml(\ystar))$ be a Nash equilibrium for the RUC game
$(A, B)$ such that $\support(\xstar) = \support(\ystar) = \toN$.
For $\delta \in [0, 1)$,
let $\xhat, \yhat \in \simplex_n$ such that for all $i \in \toN$,
$|\xhat_i - \xstar_i| \le \delta\xstar_i$ and $|\yhat_i - \ystar_i| \le \delta\ystar_i$.
Then $(\xhat, \yhat)$ is a $\frac{4\delta}{(1-\delta)^2}$-approximate Nash equilibrium for the RUC game $(A, B)$.
\end{restatable}
\begin{proof}[Proof sketch]
The proof is similar to \cref{thm:approx-ne}.
It follows from \cref{thm:meml-ne-is-eigen,thm:approx-eigen-norm,thm:batter-nmem-ne,thm:bowler-nmem-ne}.
See \cref{sec:non-meml-extra:approx} for the full proof.
\end{proof}

\subsection{Uniqueness of Nash Equilibrium up to Equivalence}
\label{sec:non-meml:equiv}

Before we investigate the uniqueness of Nash equilibrium for RUC games,
we first describe a phenomenon where two different randomized strategies can behave similarly.

\begin{definition}[collisions]
For lists $I \defeq [i_1, \ldots, i_k]$ and $J \defeq [j_1, \ldots, j_k]$,
let $\colls(I, J)$ be the number of collisions
if the \batter{} and \bowler{} play actions $I$ and $J$, respectively.
Formally, $\colls(I, J) \defeq \sum_{t=1}^k \boolone(i_t=j_t)$.

The pair $(I, J)$ is called \emph{collision-consistent} if $i_t = j_t$ for all $t < k$.
\end{definition}

\begin{definition}
Two randomized strategies $f_1$ and $f_2$ are said to be \emph{equivalent}
if for every collision-consistent pair $(I, J)$, we have
$\Pr(\isResponse(f_1, I, J)) = \Pr(\isResponse(f_2, I, J))$.
\end{definition}

\begin{example}
Let $n = 3$. For any $u, v \in \{1, 2\}$, let $f_{u,v}$ be the deterministic strategy where
$f_{u,v}(J) = u$ if $J = [2]$, $f_{u,v}(J) = v$ if $J = [3]$, and $f_{u,v}(J) = 1$ otherwise.

For any $p \in [0, 1/2]$, let $h_p$ be a randomized strategy where
$\Pr(h_p = f_{1,1}) = \Pr(h_p = f_{2,2}) = p$ and $\Pr(h_p = f_{1,2}) = \Pr(h_p = f_{2,1}) = 1/2-p$.
Then for any lists $I$ and $J$ of the same length,
$\Pr(\isResponse(h_p, I, J))$ is not a function of $p$.
Hence, $h_0$ and $h_{1/2}$ are equivalent.
\end{example}

We now show that given a pair of strategies, replacing each strategy by an equivalent
strategy makes no difference to anyone's payoff.

\begin{lemma}
\label{thm:eqv-score}
Let $f_1$ and $f_2$ be equivalent strategies of the \batter{}
and $g_1$ and $g_2$ be equivalent strategies of the \bowler{}.
Then for the RUC game $(A, B)$, we have $e^{(C,r)}(f_1, g_1) = e^{(C,r)}(f_2, g_2)$
for all $C \in \{A, B\}$ and $r \in \mathbb{Z}_{\ge 0} \cup \{\infty\}$.
\end{lemma}
\begin{proof}[Proof sketch]
For any $r \in \mathbb{Z}_{\ge 0}$, the expected score in the first $r$ rounds of a pair of strategies
only depends on the distribution of collision-consistent pairs $(I, J)$ where $|I|=|J|\le r$.
Hence, $e^{(C,r)}(f_1, g_1) = e^{(C,r)}(f_2, g_2)$.
Take the limit $r \to \infty$ to get $e^{(C,\infty)}(f_1, g_1) = e^{(C,\infty)}(f_2, g_2)$.
\end{proof}
A corollary to \cref{thm:eqv-score} is that if $(f_1, g_1)$ is a Nash equilibrium,
then so is $(f_2, g_2)$.

Next, we give a useful characterization of equivalence to stationary strategies.

\begin{restatable}{lemma}{rthmEqvMeml}
\label{thm:eqv-meml}
Let $\xvec \in \simplex_n$. A randomized strategy $f$ is equivalent to $\stdmeml(\xvec)$
iff $f^{(I,J)}(\emptyset) \sampledFrom \xvec$ for every
feasible history $(I, J)$ of $f$ where $\colls(I, J) = 0$.
\end{restatable}
\begin{proof}[Proof sketch]
(See \cref{sec:non-meml-extra:str} for the full proof.)

$\impliedby$:
For each collision-consistent pair $(I, J)$, where $I \defeq [i_1, \ldots, i_k]$, we show that
\\ $\Pr(\isResponse(f, I, J)) = \Pr(\isResponse(\stdmeml(\xvec), I, J)) = \prod_{t=1}^k \xvec_{i_t}$.

$\implies$:
$\forall i, j \in \toN$,
$\Pr(f^{(I,J)}(\emptyset) = i) = \Pr(\isResponse(f, I+[i], J+[j]))/\Pr(\isResponse(f,I,J))$.
Replace $f$ by $\stdmeml(\xvec)$ (since they're equivalent) and simplify to get
$\Pr(f^{(I,J)}(\emptyset) = i) = \xvec_i$.
\end{proof}

Indeed, for any stationary strategy, it's possible to construct a different equivalent strategy.
This rules out uniqueness of Nash equilibrium.
However, we can still hope to get uniqueness \emph{up to equivalence}.
We found this to be a very difficult problem, and could only resolve it for RUC games of the form
$(A, A)$ (i.e., zero-sum RUC games), where $A$ is irreducible.
For such games, we show that all Nash equilibria are equivalent.

\begin{theorem}
\label{thm:nec-cond-for-ne}
Let $\rho$ be the Perron root of an irreducible matrix $A$ and
$\xstar$ and $\ystar$ be the left and right Perron vectors of $A$, respectively.
Then for any Nash equilibrium $(f^*, g^*)$ of the RUC game $(A, A)$,
$f^*$ is equivalent to $\stdmeml(\xstar)$ and $g^*$ is equivalent to $\stdmeml(\ystar)$.
\end{theorem}
\begin{proof}[Proof sketch]
(See \cref{sec:non-meml-extra:nec} for the full proof.)

We first use \cref{thm:all-ne-same-value} to show that
$(f^*, \ystar)$ and $(\xstar, g^*)$ are also NE.
If $f^*(\emptyset)$ is not distributed as $\xstar$, then in the NE $(f^*, \ystar)$,
we show that the \bowler{} can decrease her cost by changing her first action's distribution.
This contradicts the fact that $(f^*, \ystar)$ is an NE,
and hence, proves that $f^*(\emptyset) \sampledFrom \xstar$.
We similarly prove that $g^*(\emptyset) \sampledFrom \ystar$.

Next, we show that for any $(I, J)$ such that $\colls(I, J) = 0$,
$({f^*}^{(I,J)}, {g^*}^{(J,I)})$ is also an NE,
otherwise the \bowler{} can gain by deviating in the NE $(f^*, \ystar)$,
or the \batter{} can gain by deviating in the NE $(\xstar, g^*)$.

Combining the above results tells us that
${f^*}^{(I,J)}(\emptyset) \sampledFrom \xstar$ and ${g^*}^{(J,I)}(\emptyset) \sampledFrom \ystar$
for all $(I, J)$ such that $\colls(I, J) = 0$,
which fulfills the condition of equivalence in \cref{thm:eqv-meml}.
\end{proof}

\section{Discussion}
\label{sec:discussion}

In this work, we initiated the study of two-player RUC games: games that are repeated until collision.
RUC games are related to other well-known repeated games, like pursuit-evasion games, hide-and-seek games, and stochastic games. They also generalize the popular game of hand cricket, and in \cref{app:hc} we discuss its popular variants.
Our main result showed the existence of Nash equilibria in RUC games when the players' payoff matrices are irreducible. We studied two other interesting properties: \textit{stationarity} and \textit{uniqueness}. We proved there always exists a Nash equilibrium where players use stationary strategies, and for zero-sum RUC games, all Nash equilibria are essentially equivalent, that is, they cannot be distinguished by observing the players' actions. 

In \cref{sec:multi-collisions}, we explore a variant of RUC games where
instead of ending the game on the first collision,
we end it on the $w\Th$ collision, for some $w \in \mathbb{Z}_{\ge 0}$.

Our work opens the way for several interesting questions. We can naturally generalize the definition of a collision to include a larger set of colliding actions, encoded via a collision matrix, and investigate the existence of Nash equilibria. Lastly, our result showing uniqueness up to equivalence of Nash equilibria applies only to zero-sum RUC games; showing (non-)uniqueness for general RUC games is another open question.

\acknowledgements{We thank Prof.~Jugal Garg and Prof.~Ruta Mehta for their helpful comments.}

\appendix
\section{Expected Value and Variance of Payoffs}
\label{sec:exp-and-var}

In this section, we show how to find the expected value and variance
of the players' total payoffs when they play stationary strategies.

\rthmExpectedScore*
\begin{proof}
In any round, let $I$ and $J$ be the numbers picked by the \batter{} and the \bowler{}, respectively.
Then $I$ and $J$ are multinoulli random variables.
The \batter{}'s score in that round is $\E(A[I,J]) = \xvec\trn A\yvec$,
and the \bowler{}'s cost in that round is $\E(B[I,J]) = \xvec\trn B\yvec$.
The probability of a collision in that round is $\Pr(I = J) = \xvec\trn\yvec$.

If $\xvec\trn\yvec = 0$, then $\xvec$ and $\yvec$ have disjoint supports,
so there will never be a collision.
If $\xvec\trn A\yvec > 0$, then the \batter{} will keep scoring, and so $e^{(A)}(\xvec, \yvec) = \infty$.
If $\xvec\trn A\yvec = 0$, then $A[i, j] = 0$ for all $i \in \support(\xvec)$
and $j \in \support(\yvec)$. Hence, the \batter{} will never score anything,
and so $e^{(A)}(\xvec, \yvec) = 0$.
We can similarly find $e^{(B)}(\xvec, \yvec)$ when $\xvec\trn\yvec = 0$.

Now assume $\xvec\trn\yvec > 0$. Then in each round, there is a positive probability of collision,
so the game will eventually end with probability 1.
Let $Z^{(A)}$ and $Z^{(B)}$ be the random variables denoting
the \batter{}'s total score and the \bowler{}'s total cost, respectively.

Let $N$ be the number of rounds the game lasts for.
Then $N$ is a geometric random variable and $\E(N) = 1/\xvec\trn\yvec$.
Let $\alpha_C \defeq \max_{i,j} C_{i,j}$.
Then $Z^{(C)} \le \alpha_CN$, so $\E(Z^{(C)})$ is finite.

Let $I$ and $J$ be the numbers picked by the \batter{} and the \bowler{}, respectively,
in the first round. Then
\begin{align*}
e^{(C)}(\xvec, \yvec) &= \E(Z^{(C)})
= \sum_{i=1}^n \sum_{j=1}^n \E(Z^{(C)} \mid I=i \textrm{ and } J=j)\Pr(I=i \textrm{ and } J=j)
\\ &= \sum_{i=1}^n \sum_{j=1}^n \xvec_i\yvec_j\left( C_{i,j} + \boolone(i\neq j)\E(Z^{(C)}) \right)
\\ &= \xvec\trn C\yvec + \E(Z^{(C)})\sum_{i=1}^n\sum_{j=1}^n \boolone(i \neq j)\xvec_i\yvec_j
\\ &= \xvec\trn C\yvec + \E(Z^{(C)})(1-\xvec\trn\yvec)
\end{align*}
Hence, $e^{(C)}(\xvec, \yvec) = \E(Z^{(C)}) = \xvec\trn C\yvec/\xvec\trn\yvec$.
\end{proof}

\begin{lemma}[Variance decomposition formula]
\label{thm:cond-var}
For any two random variables $X$ and $Y$, we have
$\Var(Y) = \Var(\E(Y \mid X)) + \E(\Var(Y \mid X))$.
\end{lemma}

\begin{lemma}
\label{thm:meml-score-var}
For the RUC game $(A, B)$, suppose the \batter{} and \bowler{} use stationary strategies
$\xvec$ and $\yvec$, respectively, such that $\xvec\trn\yvec > 0$.
Let $X$ and $Y$ be the actions picked in the first round by the \batter{} and \bowler{},
respectively (so $X$ and $Y$ are random variables).
Let $v^{(A)}(\xvec, \yvec)$ be the variance of the \batter{}'s total score
and $v^{(B)}(\xvec, \yvec)$ be the variance of the \bowler{}'s total cost.
For $C \in \{A, B\}$, let $\mu^{(C)} \defeq e^{(C)}(\xvec, \yvec)$.
Then $v^{(C)}(\xvec, \yvec) = \Var((C-\mu^{(C)}I)[X, Y])/\xvec\trn\yvec$.
\end{lemma}
\begin{proof}
We will only show how to find $v^{(A)}(\xvec, \yvec)$.
The analysis for $v^{(B)}(\xvec, \yvec)$ is analogous.

Let $N$ be the number of rounds the game lasts for.
Let $Z$ be the \batter{}'s total score.
Then $v^{(A)}(\xvec, \yvec) = \Var(Z)$ and $N$ is a geometric random variable
where $\E(N) = 1/\xvec\trn\yvec$.
Let $\alpha \defeq \max_{i,j} A_{i,j}$. Then $Z \le \alpha N$.
Hence, $\Var(Z) \le \E(Z^2) \le \alpha^2 E(N^2) = \alpha^2(2-\xvec\trn\yvec)/(\xvec\trn\yvec)^2$,
so $\Var(Z)$ is finite.

Let $Z'$ be the \batter{}'s total score if the first round is ignored,
i.e., the \batter{} gets a score of 0 in the first round
and the game doesn't end at the first round even if $X = Y$.
Then $Z = A[X,Y] + \boolone(X \neq Y)Z'$.

$Z$ and $Z'$ have the same distribution, since ignoring the first round
is akin to starting the game afresh after the first round.
Hence, $\Var(Z) = \Var(Z')$ and $\E(Z) = \E(Z')$.
Also, $Z'$ is independent of the pair $(X, Y)$.
Using \cref{thm:cond-var}, we get
$\Var(Z) = \Var(\E(Z \mid X, Y)) + \E(\Var(Z \mid X, Y))$.
\begin{align*}
\E(Z \mid X=i \land Y=j) &= \E(A[X,Y] + \boolone(X \neq Y)Z' \mid X=i \land Y=j)
\\ &= \E(A[i,j] + \boolone(i \neq j)Z' \mid X=i \land Y=j)
\\ &= A[i,j] + \boolone(i \neq j)\E(Z' \mid X=i \land Y=j)
\\ &= A[i,j] + \boolone(i \neq j)\E(Z')
    \tag{since $Z'$ is independent of $(X, Y)$}
\\ &= A[i,j] + \boolone(i \neq j)\mu^{(A)}.
\end{align*}
\begin{align*}
\Var(Z \mid X=i \land Y=j) &= \Var(A[X,Y] + \boolone(X \neq Y)Z' \mid X=i \land Y=j)
\\ &= \Var(A[i,j] + \boolone(i \neq j)Z' \mid X=i \land Y=j)
\\ &= \boolone(i \neq j)\Var(Z' \mid X=i \land Y=j)
\\ &= \boolone(i \neq j)\Var(Z')
    \tag{since $Z'$ is independent of $(X, Y)$}
\\ &= \boolone(i \neq j)v^{(A)}(\xvec, \yvec).
\end{align*}
\begin{align*}
v^{(A)}(\xvec, \yvec) = \Var(Z) &= \Var(\E(Z \mid X, Y)) + \E(\Var(Z \mid X, Y))
\\ &= \Var(A[X,Y] + \boolone(X \neq Y)\mu^{(A)}) + \E(\boolone(X \neq Y)v^{(A)}(\xvec, \yvec))
\\ &= \Var(A[X,Y] + \mu^{(A)}(1 - I[X,Y])) + v^{(A)}(\xvec, \yvec)\Pr(X \neq Y)
\\ &= \Var((A - \mu^{(A)}I)[X,Y]) + v^{(A)}(\xvec, \yvec)(1-\xvec\trn\yvec).
\end{align*}
Hence, $v^{(A)}(\xvec, \yvec) = \Var((A - \mu^{(A)}I)[X,Y])/\xvec\trn\yvec$.
\end{proof}

Using \cref{thm:meml-score-var}, we can compute the variance of the total score (or cost)
in $O(n^2)$ arithmetic operations. However, this result is not easy to interpret.
We now show how the variance can be expressed in an alternate form that gives more insight.

\begin{lemma}
\label{thm:meml-score-var-breakdown}
Let $A \in \mathbb{R}_{\ge 0}^{n \times n}$ and $\xvec, \yvec \in \simplex_n$
such that $\xvec\trn\yvec > 0$. Let $X, Y \in \toN$ be random variables such that
$X \sampledFrom \xvec$ and $Y \sampledFrom \yvec$.

Let $B$, $C$, $D$ be matrices where $B_{i,j} \defeq A_{i,j}\boolone(i \neq j)$,
$C_{i,j} \defeq B_{i,j}^2$, and $D_{i,j} \defeq A_{i,j}\boolone(i=j)$.
Let $\mu \defeq \xvec\trn A\yvec/\xvec\trn\yvec$,
$\mu_B \defeq \xvec\trn B\yvec/\xvec\trn\yvec$,
and $\mu_D \defeq \xvec\trn D\yvec/\xvec\trn\yvec$. Then
\[ \frac{\Var((A-\mu I)[X,Y])}{\xvec\trn\yvec}
= \mu_B^2 + \frac{\xvec\trn C\yvec}{\xvec\trn\yvec}
    + \frac{1}{\xvec\trn\yvec}\sum_{i=1}^n \xvec_i\yvec_i(A_{i,i}-\mu_D)^2. \]
\end{lemma}
\begin{proof}
$\E((A - \mu I)[X, Y]) = \xvec\trn(A - \mu I)\yvec = 0$.
Hence, $\Var((A-\mu I)[X,Y]) = \E((A-\mu I)[X,Y]^2)$.
Since $A = B + D$, we get $\mu = \mu_B + \mu_D$. Hence,
\[ (A-\mu I)[X,Y]^2 = (B-\mu_BI)[X,Y]^2 + (D-\mu_DI)[X,Y]^2 + 2(B-\mu_B)[X,Y](D-\mu_DI)[X,Y]. \]
Since $D - \mu_DI$ is a diagonal matrix,
$(B-\mu_BI)[X,Y](D-\mu_DI)[X,Y] = (-\mu_B)(D-\mu_DI)[X,Y]$.
Moreover, $\E((D-\mu_DI)[X,Y]) = \xvec\trn(D - \mu_DI)\yvec = 0$.
Hence, $\E((A-\mu I)[X,Y]^2) = \E((B-\mu_BI)[X,Y]^2) + \E((D-\mu_DI)[X,Y]^2)$.
\[ \E((D-\mu_DI)[X,Y]^2) = \sum_{i=1}^n \xvec_i\yvec_i(A_{i,i}-\mu_D)^2. \]
For any $i, j \in \toN$, we get
\begin{align*}
(B-\mu_BI)[i,j]^2 &= (B[i,j] - \mu_B\boolone(i=j))^2
\\ &= B[i,j]^2 + \mu_B^2\boolone(i=j) - 2\mu_B\boolone(i=j)B[i,j]
\\ &= C[i,j] + \mu_B^2\boolone(i=j).
\end{align*}
Hence, $\E((B-\mu_B)[X,Y]^2) = \xvec\trn C\yvec + \mu_B^2\xvec\trn\yvec$. Therefore,
\[ \frac{\Var((A-\mu I)[X,Y])}{\xvec\trn\yvec}
= \mu_B^2 + \frac{\xvec\trn C\yvec}{\xvec\trn\yvec}
    + \frac{1}{\xvec\trn\yvec}\sum_{i=1}^n \xvec_i\yvec_i(A_{i,i}-\mu_D)^2.
\qedhere \]
\end{proof}

\Cref{thm:meml-score-var,thm:meml-score-var-breakdown} tell us that
the variance of the total score has three components:
the first two are based on non-diagonal entries of the score matrix,
and the third is based on the diagonal entries of the score matrix.

\Cref{thm:meml-score-var,thm:meml-score-var-breakdown} also tell us that
the variance of the total score (or cost) can be quite large.
When the diagonal entries of the score matrix $A$ are zero,
then for $\mu \defeq e^{(A)}(\xvec, \yvec)$, we have $v^{(A)}(\xvec, \yvec) \ge \mu^2$.

\section{Nash Equilibrium for Hand Cricket}
\label{app:hc}
We consider two popular variants of hand cricket,
and show how to efficiently compute their (approximate) Nash equilibria.

\subsection{Variant 1}

Let $s \in \mathbb{R}^n$ such that $0 < s_1 \le s_2 \le \ldots \le s_n$.
Let $A \in \mathbb{R}_{\ge 0}^{n \times n}$ where
$A_{i,j} = s_i$ if $i \neq j$ and $A_{i,j} = 0$ if $i = j$.
Consider the RUC game $(A, A)$.

\begin{lemma}
\label{thm:std1-g}
Let $g(z) \defeq \sum_{i=1}^n \frac{s_i}{z+s_i}$. Let $r \defeq \sum_{i=1}^n s_i$.
There exists a unique $\rho \in [0, \infty)$ such that $g(\rho) = 1$.
Moreover, $g$ is strictly monotonically decreasing and $r - s_n \le \rho \le r - s_1$.
\end{lemma}
\begin{proof}
Let $0 \le z_1 < z_2$. Then
\[ g(z_1) - g(z_2) = \sum_{i=1}^n \left(\frac{s_i}{s_i+z_1} - \frac{s_i}{s_i+z_2}\right)
    = \sum_{i=1}^n \frac{s_i(z_2-z_1)}{(s_i+z_1)(s_i+z_2)} > 0. \]
Hence, $g$ is strictly monotonically decreasing in the interval $[0, \infty)$.
\begin{align*}
g(r-s_n) &= \sum_{i=1}^n \frac{s_i}{r-s_n+s_i} \ge \sum_{i=1}^n \frac{s_i}{r} = 1
\\ g(r-s_1) &= \sum_{i=1}^n \frac{s_i}{r-s_1+s_i} \le \sum_{i=1}^n \frac{s_i}{r} = 1
\end{align*}
Since $g$ is continuous, we get our desired result using the intermediate value theorem.
\end{proof}

\begin{lemma}
\label{thm:std1-xy}
Define vectors $\xstar$ and $\ystar$ as
\begin{align*}
\xstar_i &= \frac{1}{n-1}\,\frac{\rho}{\rho+s_i} \quad\forall i \in \toN
& \ystar_j &= \frac{s_j}{\rho+s_j} \quad\forall j \in \toN
\end{align*}
Then $(\xstar, \ystar)$ is a Nash equilibrium for the RUC game $(A, A)$
and $e(\xstar, \ystar) = \rho$.
\end{lemma}
\begin{proof}
\begin{align*}
\sum_{j=1}^n \ystar_j &= g(\rho) = 1.
& \sum_{i=1}^n \xstar_i &= \sum_{i=1}^n \frac{1-\ystar_i}{n-1} = 1.
\end{align*}
Hence, $\xstar, \ystar \in \simplex_n$.
\[ s\trn\xstar = \sum_{i=1}^n s_i\xstar_i = \frac{\rho}{n-1}\sum_{i=1}^n \ystar_i = \frac{\rho}{n-1}. \]
For any $j \in \toN$, we have
\[ (A\trn\xstar)_j = \sum_{i=1}^n A_{i,j}\xstar_i
    = \sum_{i \in \toN \setminus \{j\}} s_i\xstar_i = s\trn\xstar - s_j\xstar_j
    = \frac{\rho}{n-1}\left(1 - \frac{s_j}{\rho+s_j}\right) = \rho\xstar_j. \]
For any $i \in \toN$, we have
\[ (A\ystar)_i = \sum_{j=1}^n A_{i,j}\ystar_j = \sum_{j \in \toN \setminus \{i\}} s_i\ystar_j
    = s_i(1-\ystar_i) = \rho\ystar_i. \]
Hence, $(\rho, \xstar)$ is an eigenpair of $A\trnF$ and $(\rho, \ystar)$ is an eigenpair of $A$.
By \cref{thm:pos-eig-gives-ne-2}, $(\xstar, \ystar)$ is a Nash equilibrium
for the RUC game $(A, A)$, and $e(\xstar, \ystar) = \rho$.
\end{proof}

For $n=2$, solving $g(\rho) = 1$ gives $\rho = \sqrt{s_1s_2}$. Then
\begin{align*}
\xstar &= \left[\frac{\sqrt{s_2}}{\sqrt{s_1} + \sqrt{s_2}}, \frac{\sqrt{s_1}}{\sqrt{s_1} + \sqrt{s_2}} \right],
& \ystar &= \left[\frac{\sqrt{s_1}}{\sqrt{s_1} + \sqrt{s_2}}, \frac{\sqrt{s_2}}{\sqrt{s_1} + \sqrt{s_2}} \right].
\end{align*}

\subsubsection{Computing Approximate Nash Equilibrium}

We can compute $\rho$ approximately using binary search. Specifically,
for any $\eps \in (0, 1)$, we can find $\rhohat$ such that
$r-s_n \le \rhohat \le r-s_1$ and $|\rhohat-\rho| \le \eps(r-s_n)$ using at most
$\displaystyle \ceil{\log_2\left(\frac{1}{\eps}\right) + \log_2\left(\frac{s_n-s_1}{r-s_n}\right)}$
iterations of binary search, where each iteration takes $\Theta(n)$ time.
Define vectors $\uvec$ and $\vvec$ as
\begin{align*}
\uvec_i &\defeq \frac{1}{n-1}\frac{\rhohat}{\rhohat+s_i} \quad\forall i \in \toN
& \vvec_j &\defeq \frac{s_j}{\rhohat+s_j} \quad\forall j \in \toN.
\end{align*}
Let $\xhat \defeq \uvec/\|\uvec\|_1$ and $\yhat \defeq \vvec/\|\vvec\|_1$. Then
$\xhat, \yhat \in \simplex_n$, and we show that $(\xhat, \yhat)$ is an approximate Nash equilibrium.

\begin{lemma}
\label{thm:std1-approx-ne}
For any $\xvec, \yvec \in \simplex_n$,
\begin{align*}
e(\xhat, \yvec) &\ge \frac{1-3\eps}{1+\eps}\rho,
& e(\xvec, \yhat) &\le \frac{1+\eps}{1-3\eps}\rho.
\end{align*}
Moreover, $|\xhat_i - \xstar_i| \le \frac{2\eps}{1-\eps}\xstar_i$ for all $i \in \toN$
and $|\yhat_j - \ystar_j| \le \frac{2\eps}{1-\eps}\ystar_j$ for all $j \in \toN$.
\end{lemma}
\begin{proof}
Let $i \in \toN$. Then
\[ \uvec_i - \xstar_i = \frac{1}{n-1}\left(\frac{\rhohat}{\rhohat+s_i} - \frac{\rho}{\rho+s_i}\right)
= \frac{s_i}{n-1}\,\frac{\rhohat-\rho}{(\rhohat+s_i)(\rho+s_i)}. \]
\[ \implies \frac{|\uvec_i - \xstar_i|}{\xstar_i}
= \frac{s_i}{\rhohat+s_i}\,\frac{|\rhohat-\rho|}{\rho}
\le \frac{s_n}{\rhohat+s_n}\,\frac{|\rhohat-\rho|}{\rho}
\le \frac{s_n}{(r-s_n)+s_n}\,\frac{\eps(r-s_n)}{r-s_n} \le \eps\frac{s_n}{r} \le \eps. \]
\begin{align*}
&\implies \uvec_i \in [(1-\eps)\xstar_i, (1+\eps)\xstar_i]
&&\implies \|\uvec\|_1 \in [1-\eps, 1+\eps]
\\ &\implies \xhat_i = \frac{\uvec_i}{\|\uvec\|_1}
    \in \left[\frac{1-\eps}{1+\eps}, \frac{1+\eps}{1-\eps}\right]\xstar_i
&&\implies |\xhat_i - \xstar_i| \le \frac{2\eps}{1-\eps}\xstar_i.
\end{align*}
Let $j \in \toN$. Then
\[ \vvec_j - \ystar_j = \frac{s_j}{\rhohat + s_j} - \frac{s_j}{\rho + s_j}
    = \frac{s_j(\rho - \rhohat)}{(\rho+s_j)(\rhohat+s_j)}
    = \frac{\rho - \rhohat}{\rhohat + s_j}\ystar_j. \]
\[ \implies \frac{|\vvec_j - \ystar_j|}{\ystar_j} = \frac{|\rho - \rhohat|}{\rhohat+s_j}
    \le \frac{(r-s_n)\eps}{r-s_n+s_j} \le \eps. \]
\begin{align*}
&\implies \vvec_j \in [(1-\eps)\ystar_j, (1+\eps)\ystar_j]
&&\implies \|\vvec\|_1 \in [1-\eps, 1+\eps]
\\ &\implies \yhat_j = \frac{\vvec_j}{\|\vvec\|_1}
    \in \left[\frac{1-\eps}{1+\eps}, \frac{1+\eps}{1-\eps}\right]\ystar_j
&&\implies |\yhat_j - \ystar_j| \le \frac{2\eps}{1-\eps}\ystar_j.
\end{align*}
Using \cref{thm:approx-ne-nmem}, we get that for any $\xvec, \yvec \in \simplex_n$,
\begin{align*}
e(\xhat, \yvec) &\ge \frac{1-3\eps}{1+\eps}\rho,
& e(\xvec, \yhat) &\le \frac{1+\eps}{1-3\eps}\rho.
\qedhere \end{align*}
\end{proof}

\subsection{Variant 2}

Let $s \in \mathbb{R}^n$ such that $0 < s_1 \le s_2 \le \ldots \le s_n$.
Consider the RUC game $(A, A)$, where
$A \in \mathbb{R}_{\ge 0}^{(n+1) \times (n+1)}$ and
\[ A_{i,j} \defeq \begin{cases}
0 & \textrm{ if } i = j
\\ s_i & \textrm{ if } i \neq j \textrm{ and } i \in \toN
\\ s_j & \textrm{ if } i \neq j \textrm{ and } i = n+1
\end{cases}. \]

\begin{lemma}
\label{thm:std2-xy}
Let $r = \sum_{i=1}^n s_i$. Define vectors $\xstar$ and $\ystar$ as
\begin{align*}
\xstar_i &= \frac{1}{n+1}  \quad\forall i \in [n+1]
& \ystar_j &= \frac{s_j}{r+s_j} \quad\forall j \in \toN
& \ystar_{n+1} &= 1 - \sum_{j=1}^n \ystar_j
\end{align*}
Then $(\xstar, \ystar)$ is a Nash equilibrium for the RUC game $(A, A)$
and $e(\xstar, \ystar) = r$.
\end{lemma}
\begin{proof}
It's easy to see that $\xstar, \ystar \in \simplex_n$.
\[ (A\trn\xstar)_{n+1} = \sum_{i=1}^n A_{i,n+1}\xstar_i = \sum_{i=1}^n s_i\xstar_i = r\xstar_{n+1}. \]
For any $j \in \toN$, we have
\[ (A\trn\xstar)_j = \sum_{i=1}^{n+1} A_{i,j}\xstar_i
    = \sum_{i \in \toN \setminus \{j\}} s_i\xstar_i + s_j\xstar_{n+1} = r\xstar_j. \]
Hence, $(r, \xstar)$ is an eigenpair of $A\trnF$.
\[ \sum_{j=1}^n s_j\ystar_j = \sum_{j=1}^n s_j\left(1 - \frac{r}{r+s_j}\right)
    = r - r\sum_{j=1}^n \frac{s_j}{r+s_j} = r\left(1 - \sum_{j=1}^n \ystar_j \right) = r\ystar_{n+1}. \]
\[ (A\ystar)_{n+1} = \sum_{j=1}^n A_{n+1,j}\ystar_j = \sum_{j=1}^n s_j\ystar_j = r\ystar_{n+1}. \]
For any $i \in \toN$, we have
\[ (A\ystar)_i = \sum_{j=1}^{n+1} A_{i,j}\ystar_j = \sum_{j \in [n+1]-\{i\}} s_i\ystar_j
    = s_i(1-\ystar_i) = r\ystar_i. \]
Hence, $(r, \ystar)$ is an eigenpair of $A$.

By \cref{thm:pos-eig-gives-ne-2}, $(\xstar, \ystar)$ is a Nash equilibrium
and $e(\xstar, \ystar) = r$.
\end{proof}

\section{Details of General RUC Games}
\label{sec:non-meml-extra}

\subsection{Facts About Strategies}
\label{sec:non-meml-extra:str}

\begin{observation}
For any strategy $f$, we have $(f^{J_1})^{J_2} = f^{J_1+J_2}$.
\end{observation}

\begin{lemma}
\label{thm:hist-compose}
Let $f$ be a randomized strategy. Let $(I_1, J_1)$ be a feasible history for $f$ and
$(I_2, J_2)$ be a feasible history for $f^{(I_1, J_1)}$.
Then $(I_1+I_2, J_1+J_2)$ is a feasible history for $f$ and
$(f^{(I_1,J_1)})^{(I_2,J_2)}$ has the same distribution as $f^{(I_1+I_2, J_1+J_2)}$.
\end{lemma}
\begin{proof}
Suppose $f$ is sampled from the probability space $(\Omega, \Ecal, P)$.
Let $F \in \Ecal$. Let $F_0 \defeq \{\fhat \in \Omega: \fhat^{J_2} \in F\}$,
$F_1 \defeq \{\fhat \in \Omega: \isResponse(\fhat, I_1, J_1)\}$,
and $F_2 \defeq \{\fhat \in \Omega: \isResponse(\fhat, I_2, J_2)\}$. Then
\begin{align*}
\Pr((f^{(I_1,J_1)})^{(I_2,J_2)} \in F)
&= \Pr((f^{(I_1,J_1)})^{J_2} \in F \mid \Pr(f^{(I_1,J_1)} \in F_2))
\\ &= \frac{\Pr(f^{(I_1,J_1)} \in F_0 \cap F_2)}{\Pr(f^{(I_1,J_1)} \in F_2)},
\\ \Pr(f^{(I_1,J_1)} \in F_0 \cap F_2)
&= \Pr(f^{J_1} \in F_0 \cap F_2 \mid f \in F_1)
\\ &= \Pr((f^{J_1})^{J_2} \in F \textrm{ and } f^{J_1} \in F_2 \textrm{ and } f \in F_1) / \Pr(f \in F_1),
\\ \Pr(f^{(I_1,J_1)} \in F_2)
&= \Pr(f^{J_1} \in F_2 \textrm{ and } f \in F_1) / \Pr(f \in F_1).
\end{align*}
Hence,
\begin{align*}
\Pr((f^{(I_1,J_1)})^{(I_2,J_2)} \in F)
&= \Pr(f^{J_1+J_2} \in F \mid f^{J_1} \in F_2 \textrm{ and } f \in F_1)
\\ &= \Pr(f^{(I_1+I_2,J_1+J_2)} \in F).
\end{align*}
Hence, $f^{(I_1+I_2,J_1+J_2)}$ has the same distribution as $(f^{(I_1,J_1)})^{(I_2,J_2)}$.
\end{proof}

\rthmCondOnMeml*
\begin{proof}
Let $I \defeq [i_1, \ldots, i_k]$ and $J \defeq [j_1, \ldots, j_k]$.
Let $R \defeq [R_0, R_1, \ldots]$ be an infinite sequence of independent random variables
where $\forall t \in \mathbb{Z}_{\ge 0}$ and $\forall \ell \in \toN$,
$\Pr(R_t = \ell) = \xvec_{\ell}$.
Let $R'_t \defeq R_{t+k}$ and $R' \defeq [R'_0, R'_1, \ldots] = [R_k, R_{k+1}, \ldots]$.

Define $f_R$ as the strategy $f_R(K) = R_{|K|}$. Similarly define $f_{R'}$.
Now both $f_R$ and $f_{R'}$ are stationary strategies with parameter $\xvec$.
Also, $f_R^J = f_{R'}$, since for any list $K \in \toN^*$,
$f^J_R(K) = f_R(J+K) = R_{|J+K|} = R'_{|K|} = f_{R'}(K)$.

Let $F$ be a set of deterministic strategies. Then
\begin{align*}
\Pr(f_R^{(I,J)} \in F) &= \Pr(f_R^J \in F \mid \forall t \in [k], f_R([j_1, \ldots, j_{t-1}]) = i_t)
\\ &= \Pr(f_{R'} \in F \mid [R_0, \ldots, R_{k-1}] = I)
\\ &= \Pr(f_{R'} \in F).
    \tag{since $R'$ and $[R_0, \ldots, R_{k-1}]$ are independent}
\end{align*}
Hence, $f_R^{(I,J)}$ has the same distribution as $f_{R'}$.
\end{proof}

\begin{lemma}
\label{thm:cond-on-trn}
Let $f$ and $g$ be independent randomized strategies of the \batterAndBowler{}, respectively,
sampled from probability spaces $(\Omega, \Fcal, P_f)$ and $(\Omega, \Gcal, P_g)$, respectively.
Let $I \defeq [i_1, \ldots, i_k]$ and $J \defeq [j_1, \ldots, j_k]$.
Suppose $\Pr(\isResponse(f,I,J) \textrm{ and } \isResponse(g,J,I)) > 0$.
Then for any $F \in \Fcal$ and $G \in \Gcal$, we get
\[ \Pr(f^J \in F \textrm{ and } g^I \in G \mid \isResponse(f,I,J) \textrm{ and } \isResponse(g,J,I))
= \Pr(f^{(I,J)} \in F)\Pr(g^{(J,I)} \in G). \]
\end{lemma}
\begin{proof}
\begin{align*}
& \Pr(f^J \in F \land g^I \in G \mid \isResponse(f,I,J) \textrm{ and } \isResponse(g,J,I))
\\ &= \frac{\Pr(f^J \in F \land g^I \in G \land \isResponse(f, I, J) \land \isResponse(g, J, I))}%
    {\Pr(\isResponse(f, I, J) \land \isResponse(g, J, I))}
\\ &= \frac{\Pr(f^J \in F \land \isResponse(f, I, J))\Pr(g^I \in G \land \isResponse(g, J, I))}%
    {\Pr(\isResponse(f, I, J))\Pr(\isResponse(g, J, I))}
    \tag{since $f$ and $g$ are independent}
\\ &= \Pr(f^{(I,J)} \in F)\Pr(g^{(J,I)} \in G).
\qedhere
\end{align*}
\end{proof}

\rthmExpScoreNmem*
\begin{proof}
This is obvious when $r = 0$, so now let $r \ge 1$.
The second recurrence relation is obtained by setting $r \to \infty$ in the first recurrence relation.

Let $u = f(\emptyset)$ and $v = g(\emptyset)$.
Note that $\Pr(u = i \textrm{ and } v = j) = \Pr(u=i)\Pr(v=j) = \xvec_i\yvec_j$,
since $f$ and $g$ are independent and $u \sampledFrom \xvec$ and $v \sampledFrom \yvec$.
Hence,
\begin{align*}
& \E(S^{(C,r)}(f, g)) = \E\left(C[u, v] + \boolone(u \neq v)S^{(C,r-1)}(f^{[v]}, g^{[u]})\right)
\\ &\qquad= \sum_{i \in S_x}\sum_{j \in S_y} \xvec_i\yvec_j\E\left(C[u, v] +
    \boolone(u \neq v)S^{(C,r-1)}(f^{[v]}, g^{[u]}) \mid u=i \textrm{ and } v=j\right)
\\ &\qquad= \xvec\trn C\yvec + \sum_{i \in S_x}\sum_{j \in S_y} \xvec_i\yvec_j\boolone(i \neq j)
    \E\left(S^{(C,r-1)}(f^{[j]}, g^{[i]}) \mid u=i \textrm{ and } v=j\right).
\end{align*}
Note that the event $u = i$ is the same as the event $\isResponse(f,[i],[j])$,
and $v=j$ is the same as $\isResponse(g,[j],[i])$. Hence, by \cref{thm:cond-on-trn},
the joint distribution of $f^{[j]}$ and $g^{[i]}$ conditioned on $u=i$ and $v=j$
is the same as the joint distribution of $f^{([i],[j])}$ and $g^{([j],[i])}$. So,
$\E\left(S^{(C,r-1)}(f^{[j]}, g^{[i]}) \mid u=i \textrm{ and } v=j\right)
= \E\left(S^{(C,r-1)}(f^{([i],[j])}, g^{([j],[i])})\right)$.
\end{proof}

\rthmEqvMeml*
\begin{proof}
$\impliedby$:
To show that $f$ and $\stdmeml(\xvec)$ are equivalent,
we will show that for any collision-consistent pair $(I, J)$,
where $I \defeq [i_1, \ldots, i_k]$, we have
$\Pr(\isResponse(f, I, J)) = \prod_{t=1}^k \xvec_{i_t}$.

For $t \in \{0, 1, \ldots, k\}$, let
$I_t \defeq [i_1, \ldots, i_t]$, $J_t \defeq [j_1, \ldots, j_t]$,
and $\alpha_t \defeq \prod_{\ell=1}^t \xstar_{i_{\ell}}$.
We will use induction on $t$ to show that $\Pr(\isResponse(f, I_t, J_t)) = \alpha_t$
for all $t \in \{0, \ldots, k\}$.
This is trivially true for $t = 0$. Now let $t \in \toN[k]$ and assume
$\Pr(\isResponse(f, I_{t-1}, J_{t-1})) = \alpha_{t-1}$.
If $i_{\ell} \not\in \support(\xvec)$ for some $\ell \in \toN[t]$,
then $\Pr(\isResponse(f, I_t, J_t)) = \alpha_t = 0$, so now assume
$i_{\ell} \in \support(\xvec)$ for all $\ell \in \toN[t]$. Then
\begin{align*}
& \Pr(\isResponse(f, I_t, J_t))
\\ &= \Pr(\isResponse(f, I_{t-1}, J_{t-1}) \textrm{ and } f(J_{t-1}) = i_t)
\\ &= \alpha_{t-1}\Pr(f^{J_{t-1}}(\emptyset) = i_t \mid \isResponse(f, I_{t-1}, J_{t-1}))
\\ &= \alpha_{t-1}\Pr(f^{(I_{t-1}, J_{t-1})}(\emptyset) = i_t)
\\ &= \alpha_{t-1}\xvec_{i_t} = \alpha_t.
\end{align*}
Hence, by mathematical induction, we get that $\Pr(\isResponse(f, I, J)) = \alpha_k$.
Hence, $f$ is equivalent to $\stdmeml(\xvec)$.

$\implies$:
Let $(I, J)$ be any feasible history of $f$ such that $\colls(I, J) = 0$.
Let $i, j \in \toN$. Then both $(I, J)$ and $(I+[i], J+[j])$ is collision-consistent.
\begin{align*}
\Pr(f^{(I,J)}(\emptyset) = i)
&= \Pr(f(J) = i \mid \isResponse(f, I, J))
\\ &= \frac{\Pr(\isResponse(f, I+[i], J+[j]))}{\Pr(\isResponse(f, I, J))}
\\ &= \frac{\Pr(\isResponse(\stdmeml(\xvec), I+[i], J+[j]))}{\Pr(\isResponse(\stdmeml(\xvec), I, J))}
    \tag{since $f$ and $\stdmeml(\xvec)$ are equivalent}
\end{align*}
Let $I \defeq [i_1, \ldots, i_k]$.
By the previous part ($\impliedby$), we get that
$\Pr(\isResponse(\stdmeml(\xvec), I, J)) = \prod_{t=1}^k \xstar_{i_t}$
and $\Pr(\isResponse(\stdmeml(\xvec), I+[i], J+[j])) = \xvec_i \cdot \prod_{t=1}^k \xstar_{i_t}$.
Note that $\prod_{t=1}^k \xstar_{i_t} > 0$ since $(I, J)$ is a feasible history for $f$.
Hence, $\Pr(f^{(I,J)}(\emptyset) = i) = \xvec_i$ for all $i \in \toN$.
Hence, $f^{(I,J)}(\emptyset) \sampledFrom \xvec$.
\end{proof}

\subsection{Approximate Nash Equilibrium}
\label{sec:non-meml-extra:approx}

\rthmApproxNeNmem*
\begin{proof}
Let $\eps \defeq 4\delta/(1-\delta)^2$ and
$\rho_C \defeq e^{(C, \infty)}(\xstar, \ystar)$ for $C \in \{A, B\}$.
By \cref{thm:meml-ne-is-eigen}, $(\rho_B, \xstar)$ is an eigenpair of $B\trnF$.
Hence, by \cref{thm:approx-eigen-norm}, for all $j \in \toN$,
\[ \frac{(B\trn\xhat)_j}{\xhat_j} \in \left[\frac{1-\delta}{1+\delta}\rho_B,
    \frac{1+\delta}{1-\delta}\rho_B\right]. \]
By \cref{thm:batter-nmem-ne}, for any strategy $g$, we get
\[ e^{(B,\infty)}(\xhat, g) \in \bigg[\min_{j=1}^n \frac{(B\trn\xhat)_j}{\xhat_j},
        \max_{j=1}^n \frac{(B\trn\xhat)_j}{\xhat_j}\bigg]
    = \bigg[\frac{1-\delta}{1+\delta}\rho_B, \frac{1+\delta}{1-\delta}\rho_B\bigg]. \]
In particular, for $g = \stdmeml(\yhat)$, we get
$e^{(B,\infty)}(\xhat, \yhat) \le \frac{1+\delta}{1-\delta}\rho_B$.
Thus, for any strategy $g$,
\[ e^{(B,\infty)}(\xhat, g) \ge \frac{1-\delta}{1+\delta}\rho_B
    \ge \bigg(\frac{1-\delta}{1+\delta}\bigg)^2 e^{(B,\infty)}(\xhat, \yhat)
    = \frac{e^{(B,\infty)}(\xhat,\yhat)}{1+\eps}. \]
Similarly, using \cref{thm:meml-ne-is-eigen,thm:approx-eigen-norm,thm:bowler-nmem-ne}, we can show that
\[ e^{(A,\infty)}(f, \yhat) \le \frac{1+\delta}{1-\delta}\rho_A
    \le \bigg(\frac{1+\delta}{1-\delta}\bigg)^2 e^{(A,\infty)}(\xhat, \yhat)
    = (1+\eps)e^{(A,\infty)}(\xhat, \yhat). \]
Hence, $(\xhat, \yhat)$ is an $\eps$-approximate Nash equilibrium.
\end{proof}

\subsection{Proof of Uniqueness of NE up to Equivalence}
\label{sec:non-meml-extra:nec}

We consider the RUC game $(A, A)$, where $A$ is an irreducible matrix.
Let $\rho$ be the Perron root of $A$, and
$\xstar$ and $\ystar$ be the left and right Perron vectors of $A$, respectively.
We will show that for every Nash equilibrium $(f^*, g^*)$ of the RUC game $(A, A)$,
$f^*$ is equivalent to $\stdmeml(\xstar)$ and $g^*$ is equivalent to $\stdmeml(\ystar)$.

For notational convenience, we write $S^{(r)}(f, g)$ and $e^{(r)}(f, g)$ instead of
$S^{(A, r)}(f, g)$ and $e^{(A, r)}(f, g)$.
If $(f^*, g^*)$ is any Nash equilibrium for the RUC game $(A, A)$,
then by \cref{thm:all-ne-same-value}, we get
$e^{(\infty)}(f^*, g^*) = e^{(\infty)}(f^*, \ystar) = e^{(\infty)}(\xstar, g^*) = \rho$,
and that $(\xstar, g^*)$ and $(f^*, \ystar)$ are also Nash equilibria.

We first show (in \cref{thm:xhat-equals-xstar,thm:yhat-equals-ystar})
that if $(f^*, g^*)$ is a Nash equilibrium for the RUC game $(A, A)$,
then $f^*(\emptyset) \sampledFrom \xstar$ and $g^*(\emptyset) \sampledFrom \ystar$.

\begin{lemma}
\label{thm:irred-disj-supp-prod}
Let $A \in \mathbb{R}_{\ge 0}^{n \times n}$ be an irreducible matrix.
Let $\xvec, \yvec \in \simplex_n$ such that $\support(\xvec) = \toN \setminus \support(\yvec)$.
Then $\xvec\trn A\yvec > 0$.
\end{lemma}
\begin{proof}
Let $S_x \defeq \support(\xvec)$ and $S_y \defeq \support(\yvec)$.
Since $A$ is irreducible, there is an edge $(u, v)$ in $\graph(A)$
from $S_x$ to $S_y$ such that $A_{u,v} > 0$. Hence,
$\xvec\trn A\yvec = \sum_{i \in S_x}\sum_{j \in S_y}A_{i,j}\xvec_i\yvec_j \ge A_{u,v}\xvec_u\yvec_v > 0$.
\end{proof}

\begin{lemma}
\label{thm:yhat-equals-ystar}
Let $(f^*, g^*)$ be any Nash equilibrium for the RUC game $(A, A)$.
Then $g^*(\emptyset) \sampledFrom \ystar$.
\end{lemma}
\begin{proof}
Let $g^*(\emptyset) \sampledFrom \yhat$.
Consider the Nash equilibrium $(\xstar, g^*)$.
If $\yhat \neq \ystar$, then perhaps the \batter{} can improve by changing her first action,
which would contradict the fact that $(\xstar, g^*)$ is a Nash equilibrium.

Let $\xvec \in \simplex_n$. Let $\fxvec$ be a strategy where
$\fxvec(\emptyset) \sampledFrom \xvec$ and
$\fxvec^{[j]} = \stdmeml(\xstar)$ for all $j \in \toN$.

Since $(\xstar, g^*)$ is a Nash equilibrium, the \batter{} cannot improve by
switching to $\fxvec$ for any $\xvec \in \simplex_n$.
Hence, $e^{(\infty)}(\fxvec, g^*) \le e^{(\infty)}(\xstar, g^*) = \rho$.
Let $S_x \defeq \support(\xvec)$ and $S_y \defeq \support(\yhat)$. Then
\begin{align*}
\rho &\ge e^{(\infty)}(\fxvec, g^*)
\\ &= \xvec\trn A\yhat + \sum_{i \in S_x}\sum_{j \in S_y}
    \xvec_i\yhat_j \boolone(i \neq j)e^{(\infty)}(\xstar, {g^*}^{([j],[i])})
\\ &\ge \xvec\trn A\yhat + \sum_{i=1}^n\sum_{j=1}^n \xvec_i\yhat_j\boolone(i \neq j)\rho
    \tag{by \cref{thm:ne-exists-2}}
\\ &= \xvec\trn A\yhat + \rho - \rho\xvec\trn\yhat.
\end{align*}
Hence, $\xvec\trn A\yhat \le \rho\xvec\trn\yhat$ for all $\xvec \in \simplex_n$.
Suppose $S_y \neq \toN$. Then for any $\xvec$ such that $S_x = \toN \setminus S_y$,
we get $\xvec\trn\yhat = 0$ and $\xvec\trn A\yhat > 0$ by \cref{thm:irred-disj-supp-prod}
(since $A$ is irreducible), which is a contradiction.
Hence, $S_y = \toN$. Therefore, $\xvec\trn\yhat > 0$ for all $\xvec \in \simplex_n$.

Hence, $e^{(\infty)}(\xvec, \yhat) = \xvec\trn A\yhat/\xvec\trn\yhat \le \rho$.
By \cref{thm:batter-nmem-ne}, $e^{(\infty)}(\xstar, \yvec) = \rho$ for all $\yvec \in \simplex_n$.
Hence, $(\xstar, \yhat)$ is a Nash equilibrium for the SRUC game $(A, A)$.
By \cref{thm:meml-ne-unique}, $(\xstar, \ystar)$ is the unique Nash equilibrium
for the SRUC game $(A, A)$. Hence, $\yhat = \ystar$.
\end{proof}

\begin{lemma}
\label{thm:xhat-equals-xstar}
Let $(f^*, g^*)$ be any Nash equilibrium for the RUC game $(A, A)$.
Then $f^*(\emptyset) \sampledFrom \xstar$.
\end{lemma}
\begin{proof}
Let $f^*(\emptyset) \sampledFrom \xhat$.
Consider the Nash equilibrium $(f^*, \ystar)$.
If $\xhat \neq \xstar$, then perhaps the \bowler{} can improve by changing her first action,
which would contradict the fact that $(f^*, \ystar)$ is a Nash equilibrium.

Let $\yvec \in \simplex_n$. Let $\gyvec$ be a strategy where
$\gyvec(\emptyset) \sampledFrom \yvec$
and $\gyvec^{[i]} = \stdmeml(\ystar)$ for all $i \in \toN$.

Since $(f^*, \ystar)$ is a Nash equilibrium, the \bowler{} cannot improve by
switching to $\gyvec$ for any $\yvec \in \simplex_n$.
Hence, $e^{(\infty)}(f^*, \gyvec) \ge e^{(\infty)}(f^*, \ystar) = \rho$.
Let $S_x \defeq \support(\xhat)$ and $S_y \defeq \support(\yvec)$. Then
\begin{align*}
\rho &\le e^{(\infty)}(f^*, \gyvec)
\\ &= \xhat\trn A\yvec + \sum_{i \in S_x}\sum_{j \in S_y}
    \xhat_i\yvec_j \boolone(i \neq j)e^{(\infty)}({f^*}^{([i],[j])}, \ystar)
\\ &\le \xhat\trn A\yvec + \sum_{i=1}^n\sum_{j=1}^n \xhat_i\yvec_j\boolone(i \neq j)\rho
    \tag{by \cref{thm:ne-exists-2}}
\\ &= \xhat\trn A\yvec + \rho - \rho\xhat\trn\yvec.
\end{align*}
Hence, $\xhat\trn A\yvec \ge \rho\xhat\trn\yvec$ for all $\yvec \in \simplex_n$.

Suppose $S_x = \toN$. Then $\xhat\trn\yvec > 0$ for all $\yvec \in \simplex_n$
and $e^{(\infty)}(\xhat, \yvec) = \xhat\trn A\yvec/\xhat\trn\yvec \ge \rho$.
Also, $e^{(\infty)}(\xvec, \ystar) = \rho$ for all $\xvec \in \simplex_n$ by \cref{thm:ne-exists}.
Hence, $(\xhat, \ystar)$ is a Nash equilibrium for the SRUC game $(A, A)$.
By \cref{thm:meml-ne-unique}, $(\xstar, \ystar)$ is the unique
Nash equilibrium for the SRUC game $(A, A)$, so $\xhat = \xstar$.

Now let's consider the case $S_x \neq \toN$.
Let $\alpha \defeq \sum_{j \in S_x} \ystar_j$.
Since $\support(\ystar) = \toN$, we get $0 < \alpha < 1$.
Define vectors $\yhat$ and $\ytild$ as follows:
\begin{align*}
\yhat_j &= \begin{cases}\ystar_j/\alpha & \textrm{ if } j \in S_x
\\ 0 & \textrm{ if } j \not\in S_x\end{cases},
& \ytild_j &= \begin{cases}0 & \textrm{ if } j \in S_x
\\ \ystar_j/(1-\alpha) & \textrm{ if } j \not\in S_x\end{cases}.
\end{align*}
Then $\ystar = \alpha\yhat + (1-\alpha)\ytild$ and $\xhat\trn\ystar = \alpha\xhat\trn\yhat$.
Also, $\xhat\trn A\ytild > 0$ by \cref{thm:irred-disj-supp-prod}. Hence,
\[ \rho\alpha\xhat\trn\yhat = \rho\xhat\trn\ystar = \xhat\trn A\ystar
    = \alpha\xhat\trn A\yhat + (1-\alpha)\xhat\trn A\ytild > \alpha\xhat\trn A\yhat. \]
Hence, $\xhat\trn A\yhat < \rho\xhat\trn\yhat$.
This is a contradiction. Hence, it's impossible for $S_x \neq \toN$, and so $\xhat = \xstar$.
\end{proof}

Next, we show (in \cref{thm:subg-nopt,thm:subf-nopt,thm:sub-ne,thm:rec-sub-ne}) that
residual strategies of Nash equilibria are themselves Nash equilibria.

\begin{lemma}
\label{thm:subg-nopt}
Let $(f^*, g^*)$ be any Nash equilibrium for the RUC game $(A, A)$.
Then for any $i, j \in \toN$ such that $i \neq j$ and any strategy $f$,
$e^{(\infty)}(f, {g^*}^{([j],[i])}) \le \rho$.
\end{lemma}
\begin{proof}
Suppose this is false. Then $\exists \ihat, \jhat \in \toN$ such that $\ihat \neq \jhat$
and $\exists \fhat$ such that $\delta \defeq e^{(\infty)}(\fhat, {g^*}^{([\jhat],[\ihat])}) - \rho > 0$.
Let $K \defeq \toN^2 \setminus \{(\ihat, \jhat)\}$.

For the Nash equilibrium $(\xstar, g^*)$, the \batter{} can perhaps improve
by switching to $\fhat$ if the first pair of actions is $(\ihat, \jhat)$.

Let $\ftild$ be the strategy where
$\ftild(\emptyset) \sampledFrom \xstar$ and
\[ \ftild^{[j]} \defeq \begin{cases}
\fhat & \textrm{ if } \ftild(\emptyset)=\ihat \textrm{ and } j=\jhat
\\ \stdmeml(\xstar) & \textrm{ otherwise} \end{cases}. \]
Note that $\ftild^{([i],[j])}$ is distributed as $\fhat$ if $i = \ihat$ and $j = \jhat$,
and as $\stdmeml(\xstar)$ otherwise.

Since $(\xstar, g^*)$ is a Nash equilibrium, $e^{(\infty)}(\ftild, g^*) \le e^{(\infty)}(\xstar, g^*) = \rho$.
By \cref{thm:yhat-equals-ystar}, $g^*(\emptyset) \sampledFrom \ystar$. Hence,
\begin{align*}
\rho &\ge e^{(\infty)}(\ftild, g^*)
= \xstarT A\ystar + \sum_{i=1}^n\sum_{j=1}^n \xstar_i\ystar_j
    \boolone(i \neq j)e^{(\infty)}(\ftild^{([i],[j])}, {g^*}^{([j],[i])})
\\ &= \xstarT(\rho\ystar)
    + \sum_{(i,j) \in K} \xstar_i\ystar_j\boolone(i \neq j)e^{(\infty)}(\xstar, {g^*}^{([j],[i])})
    + \xstar_{\ihat}\ystar_{\jhat}e^{(\infty)}(\fhat, {g^*}^{([\jhat],[\ihat])})
\\ &\ge \rho\xstarT\ystar + \sum_{(i,j) \in K} \xstar_i\ystar_j\boolone(i \neq j)\rho
    + \xstar_{\ihat}\ystar_{\jhat}(\rho + \delta)
    \tag{by \cref{thm:ne-exists-2}}
\\ &= \rho\xstarT\ystar + (\rho - \rho\xstarT\ystar) + \xstar_{\ihat}\ystar_{\jhat}\delta
\\ &= \rho + \xstar_{\ihat}\ystar_{\jhat}\delta > \rho.
\end{align*}
This is a contradiction. Hence, $\delta \le 0$, and so
for any $i, j \in \toN$ such that $i \neq j$ and any strategy $f$,
we have $e^{(\infty)}(f, {g^*}^{([j],[i])}) \le \rho$.
\end{proof}

\begin{lemma}
\label{thm:subf-nopt}
Let $(f^*, g^*)$ be any Nash equilibrium for the RUC game $(A, A)$.
Then for any $i, j \in \toN$ such that $i \neq j$ and any strategy $g$,
$e^{(\infty)}({f^*}^{([i],[j])}, g) \ge \rho$.
\end{lemma}
\begin{proof}
Suppose this is false. Then $\exists \ihat, \jhat \in \toN$ such that $\ihat \neq \jhat$
and $\exists \ghat$ such that $\delta \defeq \rho - e^{(\infty)}({f^*}^{([\ihat],[\jhat])}, \ghat) > 0$.
Let $K \defeq \toN^2 \setminus \{(\ihat, \jhat)\}$.

For the Nash equilibrium $(f^*, \ystar)$, the \bowler{} can perhaps improve
by switching to $\ghat$ if the first pair of actions is $(\ihat, \jhat)$.

Let $\gtild$ be the strategy where $\gtild(\emptyset) \sampledFrom \yvec$ and
\[ \gtild^{[i]} \defeq \begin{cases}
\ghat & \textrm{ if } i=\ihat \textrm{ and } \gtild(\emptyset)=\jhat
\\ \stdmeml(\ystar) & \textrm{ otherwise} \end{cases}. \]
Note that $\gtild^{([j],[i])}$ is distributed as $\ghat$ if $i=\ihat$ and $j=\jhat$,
and as $\stdmeml(\ystar)$ otherwise.

Since $(f^*, \ystar)$ is a Nash equilibrium, $e^{(\infty)}(f^*, \gtild) \ge e^{(\infty)}(f^*, \ystar) = \rho$.
By \cref{thm:xhat-equals-xstar}, $f^*(\emptyset) \sampledFrom \xstar$. Hence,
\begin{align*}
\rho &\le e^{(\infty)}(f^*, \gtild)
= \xstarT A\ystar + \sum_{i=1}^n\sum_{j=1}^n \xstar_i\ystar_j
    \boolone(i \neq j)e^{(\infty)}({f^*}^{([i],[j])}, \gtild^{([j],[i])})
\\ &= \xstarT(\rho\ystar) + \sum_{(i,j) \in K} \xstar_i\ystar_j\boolone(i \neq j)e^{(\infty)}({f^*}^{([i],[j])}, \ystar)
    + \xstar_{\ihat}\ystar_{\jhat}e^{(\infty)}({f^*}^{([\ihat],[\jhat])}, \ghat)
\\ &\le \rho\xstarT\ystar + \sum_{(i,j) \in K} \xstar_i\ystar_j\boolone(i \neq j)\rho
    + \xstar_{\ihat}\ystar_{\jhat}(\rho - \delta)
    \tag{by \cref{thm:ne-exists-2}}
\\ &= \rho\xstarT\ystar + (\rho - \rho\xstarT\ystar) - \xstar_{\ihat}\ystar_{\jhat}\delta
\\ &= \rho - \xstar_{\ihat}\ystar_{\jhat}\delta < \rho.
\end{align*}
This is a contradiction. Hence, $\delta \le 0$, and so
for any $i, j \in \toN$ such that $i \neq j$ and any strategy $g$,
we have $e^{(\infty)}({f^*}^{([i],[j])}, g) \ge \rho$.
\end{proof}

\begin{lemma}
\label{thm:sub-ne}
Let $(f^*, g^*)$ be any Nash equilibrium for the RUC game $(A, A)$.
Then for any $i, j \in \toN$ such that $i \neq j$, $({f^*}^{([i],[j])}, {g^*}^{([j],[i])})$ is
also a Nash equilibrium for the RUC game $(A, A)$.
(Hence, $e^{(\infty)}({f^*}^{([i],[j])}, {g^*}^{([j],[i])}) = \rho$.)
\end{lemma}
\begin{proof}
Follows from \cref{thm:subg-nopt,thm:subf-nopt}.
\end{proof}

\begin{lemma}
\label{thm:rec-sub-ne}
Let $(f^*, g^*)$ be any Nash equilibrium for the RUC game $(A, A)$.
Then for any lists $I$ and $J$ of the same length such that $\colls(I, J) = 0$, we get that
$({f^*}^{(I,J)}, {g^*}^{(J,I)})$ is also a Nash equilibrium for the RUC game $(A, A)$.
\end{lemma}
\begin{proof}
For lists $I$ and $J$ of the same length such that $\colls(I, J) = 0$,
let $P(I,J)$ be the proposition that $({f^*}^{(I,J)}, {g^*}^{(J,I)})$ is a Nash equilibrium
for the RUC game $(A, A)$.
We will prove $P(I,J)$ for all $I$ and $J$ by induction on $|I|$.
When $|I| = |J| = 0$ (base case), then $P(I,J)$ is trivially true.
Now let $|I| \ge 1$, $I = I' + [\ihat]$, $J = J' + [\jhat]$.
Since $\colls(I, J) = 0$, we get $\colls(I', J') = 0$ and $\ihat \neq \jhat$.
Assume $P(I',J')$ holds.

$P(I',J')$ implies $({f^*}^{(I',J')}, {g^*}^{(J',I')})$ is a Nash equilibrium for the RUC game $(A, A)$.
Since $\colls(I',J') = 0$ and $\ihat \neq \jhat$, by \cref{thm:sub-ne},
$({{f^*}^{(I',J')}}^{([\ihat],[\jhat])}, {{g^*}^{(J',I')}}^{([\jhat],[\ihat])})$
is a Nash equilibrium for the RUC game $(A, A)$.
By \cref{thm:hist-compose},
${{f^*}^{(I',J')}}^{([\ihat],[\jhat])}$ and ${{g^*}^{(J',I')}}^{([\jhat],[\ihat])}$
have the same distribution as ${f^*}^{(I,J)}$ and ${g^*}^{(J,I)}$, respectively.
So, $({f^*}^{(I,J)}, {g^*}^{(J,I)})$ is a Nash equilibrium for the RUC game $(A, A)$.
Hence, by mathematical induction, we get that $P(I,J)$ is true for all $I$ and $J$
of the same length such that $\colls(I, J) = 0$.
\end{proof}

Let $(f^*, g^*)$ be any Nash equilibrium of the RUC game $(A, A)$.
For any lists $I$ and $J$ such that $\colls(I, J) = 0$, we get that
$({f^*}^{(I,J)}, {g^*}^{(J,I)})$ is also a Nash equilibrium by \cref{thm:rec-sub-ne}.
By \cref{thm:xhat-equals-xstar,thm:yhat-equals-ystar}, we get that
${f^*}^{(I,J)}(\emptyset) \sampledFrom \xstar$ and ${g^*}^{(J,I)}(\emptyset) \sampledFrom \ystar$.
Hence, by \cref{thm:eqv-meml}, $f^*$ is equivalent to $\stdmeml(\xstar)$,
and $g^*$ is equivalent to $\stdmeml(\ystar)$.

\section{Multiple Collisions}
\label{sec:multi-collisions}

We can extend RUC games to a setting where instead of ending the game on the first collision,
we end it on the $w\Th$ collision, for some $w \in \mathbb{Z}_{\ge 0}$.
Here $w$ is called the \emph{collision threshold}.
($w = 0$ means that the game has 0 rounds and both players get a payoff of 0.)
This extension is important in hand cricket when the batter has \emph{multiple wickets}.
We represent this extension of RUC games as a tuple $(A, B, w)$.
In \cref{sec:meml,sec:non-meml}, we studied the game $(A, B, 1)$,
and now we would like to extend those results to $(A, B, w)$.

But before that, let us further extend RUC games to a setting where
the collision threshold can be a random variable $W$.
In the drug dealer example, suppose whenever the drug dealer and law enforcement
pick the same location, the drug dealer is caught with probability $p \in (0, 1]$.
Then $W$ is a geometric random variable with $\E(W) = 1/p$.

Formally, we represent this extension as a tuple $(A, B, \Dcal)$.
Here $A$ and $B$ are the payoff matrices as before,
and $\Dcal$ is a probability distribution over $\mathbb{Z}_{\ge 0}$.
The collision threshold $W$ is sampled from distribution $\Dcal$ before the game begins,
and the game runs until $W$ collisions happen.
The players know $\Dcal$ but not $W$; they find out the value of $W$
only when the $W\Th$ collision happens and the game ends.

We show that most of our results can be extended to this new setting.

\subsection{Stationary RUC Games}

We start by analyzing stationary RUC games, i.e., games where the players
are forced to only use stationary strategies.

When the \batter{} and \bowler{} use stationary strategies $\xvec$ and $\yvec$, respectively,
let $e^{(A)}_w(\xvec, \yvec)$ and $e^{(B)}_w(\xvec, \yvec)$ be the
\batter{}'s expected total score and the \bowler{}'s expected total cost, respectively,
for the game $(A, B, w)$.
Similarly, let $v^{(A)}_w(\xvec, \yvec)$ and $v^{(B)}_w(\xvec, \yvec)$ denote
the variance of the \batter{}'s total score and \bowler{}'s total cost, respectively,
for the game $(A, B, w)$.

\begin{lemma}
\label{thm:score-ev}
Let $\xvec, \yvec \in \simplex_n$ and $w \in \mathbb{Z}_{\ge 0}$.
Then for $C \in \{A, B\}$, we get
$e^{(C)}_w(\xvec, \yvec) = w e^{(C)}_1(\xvec, \yvec)$,
and when $\xvec\trn\yvec > 0$, we get
$v^{(C)}_w(\xvec, \yvec) = w v^{(C)}_1(\xvec, \yvec)$.
\end{lemma}
\begin{proof}
If $\xvec\trn\yvec = 0$, then $\xvec$ and $\yvec$ have disjoint supports, so a collision never happens.
Hence, $e^{(C)}_w(\xvec, \yvec)$ doesn't depend on $w$ and is either $0$ or $\infty$.
Hence, $e^{(C)}_w(\xvec, \yvec) = w e^{(C)}_1(\xvec, \yvec)$.

Now assume $\xvec\trn\yvec > 0$.
When $w = 0$, then this lemma is trivially true, so let $w \ge 1$.

Let $N$ be the round in which the first collision happens.
Then $N$ is a geometric random variable and $\E(N) = 1/\xvec\trn\yvec$.
Let $Z^{(A)}$ be the \batter{}'s total score in the first $N$ rounds,
and $Y^{(A)}$ be the \batter{}'s total score in the remaining game (i.e., after the first $N$ rounds).
Similarly define $Z^{(B)}$ and $Y^{(B)}$.

After the first $N$ rounds, the remaining game is the same as $(A, B, w-1)$. Hence,
\[ e^{(C)}_w(\xvec, \yvec) = \E(Z^{(C)} + Y^{(C)}) = \E(Z^{(C)}) + \E(Y^{(C)})
= e^{(C)}_1(\xvec, \yvec) + e^{(C)}_{w-1}(\xvec, \yvec). \]
Furthermore, since strategies are stationary, $Y^{(C)}$ is independent of $Z^{(C)}$. Hence,
\[ v^{(C)}_w(\xvec, \yvec) = \Var(Z^{(C)} + Y^{(C)}) = \Var(Z^{(C)}) + \Var(Y^{(C)})
= v^{(C)}_1(\xvec, \yvec) + v^{(C)}_{w-1}(\xvec, \yvec). \]
Now solve these recurrence relations in $w$ to get the intended result.
\end{proof}

\Cref{thm:score-ev} says that changing the collision threshold
is equivalent to scaling the payoffs.
Hence, for the SRUC game $(A, B, \Dcal)$, all results of \cref{sec:meml} continue to hold
when $\Dcal$'s mean is positive and finite.

\subsection{General RUC Games}

For general RUC games (where players are allowed to use non-stationary strategies),
extending the results of \cref{sec:non-meml} would take more work,
since players may now use strategies that take into account the number
of collisions that have happened so far.

We show how to extend the definitions and theorem statements of \cref{sec:non-meml}
to the multiple-collisions setting. We omit the proofs of many theorems because
they are very similar to those in \cref{sec:non-meml,sec:non-meml-extra}.
We extend our existence results to randomized collision thresholds
and our uniqueness-up-to-equivalence results to deterministic collision thresholds.

The formal definition of deterministic and randomized strategies remains the same as before.
However, the score (or cost) of a pair of strategies now needs to take the
collision threshold into account.

\begin{definition}[score]
For the RUC game $(A, B, w)$, let $S^{(A, r)}_w(f, g)$ and $S^{(B, r)}_w(f, g)$
be the \batter{}'s total score and \bowler{}'s total cost, respectively, in the first $r$ rounds
when the \batter{} uses strategy $f$ and the \bowler{} uses strategy $g$.
For $C \in \{A, B\}$, let $e^{(C, r)}_w(f, g) \defeq \E(S^{(C, r)}_w(f, g))$
and $e^{(C,\infty)}_w(f, g) \defeq \lim_{r \to \infty} e^{(C,r)}_w(f, g)$.
\end{definition}

Note that when the collision threshold $W$ is random, then $e^{(C,r)}_W$ is also random.
The agents' expected payoffs at the strategy pair $(f, g)$ would thus be
$\E(e^{(C,r)}_W(f, g))$ for $C \in \{A, B\}$.

\begin{observation}
\label{thm:mult:score-nmem}
Let $f$ and $g$ be strategies of the \batterAndBowler{}, respectively.
Let $u \defeq f(\emptyset)$ and $v \defeq g(\emptyset)$. Then for $C \in \{A, B\}$,
\[ S^{(C,r)}_w(f, g) = \begin{cases}
C[u, v] + S^{(C,r-1)}_{w-\boolone(i=j)}(f^{[v]}, g^{[u]}) & \textrm{ if } r > 0 \textrm{ and } w > 0
\\ 0 & \textrm{ otherwise}\end{cases}. \]
\end{observation}

\begin{lemma}
\label{thm:mult:expected-score-nmem}
Let $f$ and $g$ be independent strategies of the \batterAndBowler{}, respectively.
Let $f(\emptyset) \sampledFrom \xvec$, $g(\emptyset) \sampledFrom \yvec$,
$S_x \defeq \support(\xvec)$, and $S_y \defeq \support(\yvec)$. Then for $C \in \{A, B\}$,
\[ e^{(C,r)}_w(f, g) = \begin{cases}
\displaystyle \xvec\trn C\yvec + \!\sum_{i \in S_x}\sum_{j \in S_y}\xvec_i\yvec_j
    e^{(C,r-1)}_{w-\boolone(i=j)}(f^{([i],[j])}, g^{([j],[i])})
    & r \in \mathbb{Z}_{> 0} \cup \{\infty\}, w > 0
\\ 0 & r = 0 \textrm{ or } w = 0 \end{cases}. \]
\end{lemma}

\subsubsection{Existence of Nash Equilibrium}

To see whether a pair of stationary strategies can give us a Nash equilibrium,
we first investigate (in the next two lemmas) the upper and lower bounds on a player's payoff
when she is free to play any strategy (even non-stationary ones)
and her opponent uses a stationary strategy with full support.

\begin{lemma}
\label{thm:mult:batter-nmem-ne}
Let $(A, B, w)$ be an RUC game. Let $\xhat \in \simplex_n$ such that $\support(\xhat) = \toN$.
Let $\alpha \defeq \min_{i=1}^n (B\trn\xhat)_i/\xhat_i$
and $\beta \defeq \max_{i=1}^n (B\trn\xhat)_i/\xhat_i$.
Then for any strategy $g$, we have $w\alpha \le e^{(B,\infty)}_w(\xhat, g) \le w\beta$.
\end{lemma}

\begin{lemma}
\label{thm:mult:bowler-nmem-ne}
Let $(A, B, w)$ be an RUC game. Let $\yhat \in \simplex_n$ such that $\support(\yhat) = \toN$.
Let $\alpha \defeq \min_{i=1}^n (A\yhat)_i/\yhat_i$
and $\beta \defeq \max_{i=1}^n (A\yhat)_i/\yhat_i$.
Then for any strategy $f$, we have $w\alpha \le e^{(A,\infty)}_w(f, \yhat) \le w\beta$.
\end{lemma}

Next, we identify sufficient conditions to make the upper and lower bounds in
\cref{thm:mult:batter-nmem-ne,thm:mult:bowler-nmem-ne} coincide,
which gives us a Nash equilibrium.

\begin{lemma}
\label{thm:mult:pos-eig-gives-ne-2}
Let $A, B \in \mathbb{R}_{\ge 0}^{n \times n}$.
Let $(\alpha, \ystar)$ be an eigenpair for $A$ such that
$\|\ystar\|_1 = 1$ and $\ystar_i > 0$ for all $i \in \toN$.
Let $(\beta, \xstar)$ be an eigenpair for $B\trnF$ such that
$\|\xstar\|_1 = 1$ and $\xstar_i > 0$ for all $i \in \toN$.
Then for any distribution $\Dcal$, $(\xstar, \ystar)$ is a Nash equilibrium
for the RUC game $(A, B, \Dcal)$.
Moreover, for any strategies $f$ and $g$, and any $w \ge 0$, we have
$e^{(A,\infty)}_w(f, \ystar) = w\alpha$ and $e^{(B,\infty)}_w(\xstar, g) = w\beta$.
\end{lemma}
\begin{proof}
For all $w \ge 0$, and any strategies $f$ and $g$, we get
$e^{(A,\infty)}_w(f, \ystar) = w\alpha$ and $e^{(B,\infty)}_w(\xstar, g) = w\beta$
by \cref{thm:mult:batter-nmem-ne,thm:mult:bowler-nmem-ne}.

Suppose $W$ is drawn from distribution $\Dcal$. Then for any strategies $f$ and $g$,
$e^{(A,\infty)}_W(f, \ystar) = e^{(A,\infty)}_W(\xstar, \ystar)$
and $e^{(B,\infty)}_W(\xstar, g) = e^{(B,\infty)}_W(\xstar, \ystar)$.
Hence, $(\xstar, \ystar)$ is a Nash equilibrium for the RUC game $(A, B, \Dcal)$.
\end{proof}

Next, we show that conditions of \cref{thm:mult:pos-eig-gives-ne-2}
can be satisfied using the Perron-Frobenius theorem (\cref{thm:perron-frobenius}),
so a Nash equilibrium given by stationary strategies always exists.

\begin{theorem}
\label{thm:mult:ne-exists-2}
For irreducible matrices $A$ and $B$,
let $(\rho_A, \rho_B, \xstar, \ystar) = \perronSolution(A, B)$.
Then for any distribution $\Dcal$, $(\xstar, \ystar)$ is a Nash equilibrium
for the RUC game $(A, B, \Dcal)$.
Moreover, for any strategies $f$ and $g$, and any $w \ge 0$, we have
$e^{(A,\infty)}_w(f, \ystar) = w\rho_A$ and $e^{(B,\infty)}_w(\xstar, g) = w\rho_B$.
\end{theorem}
\begin{proof}
Follows from \cref{thm:mult:pos-eig-gives-ne-2}.
\end{proof}

Next, we show that the set of full support stationary Nash equilibria is the same for
different games (SRUC vs RUC, and different collision thresholds).
Hence, when we refer to a stationary Nash equilibrium,
we need not specify whether it's for an SRUC game or an RUC game,
and we need not specify the collision threshold.

\begin{lemma}
\label{thm:mult:ne-equiv}
Let $A, B \in \mathbb{R}_{\ge 0}^{n \times n}$.
Let $\xstar, \ystar \in \simplex_n$ such that $\support(\xstar) = \support(\ystar) = \toN$.
Let $\Dcal$ be a distribution over $\mathbb{Z}_{\ge 0}$ such that
$\Dcal$'s mean is positive and finite.
Then the following statements are equivalent:
\begin{enumerate}
\item \label{item:ne-equiv:sd}$(\xstar, \ystar)$ is a Nash equilibrium for the SRUC game $(A, B, \Dcal)$.
\item \label{item:ne-equiv:s1}$(\xstar, \ystar)$ is a Nash equilibrium for the SRUC game $(A, B, 1)$.
\item \label{item:ne-equiv:d}$(\xstar, \ystar)$ is a Nash equilibrium for the RUC game $(A, B, \Dcal)$.
\item \label{item:ne-equiv:1}$(\xstar, \ystar)$ is a Nash equilibrium for the RUC game $(A, B, 1)$.
\end{enumerate}
\end{lemma}
\begin{proof}
\ref{item:ne-equiv:sd} and \ref{item:ne-equiv:s1} are equivalent since
switching from the SRUC game $(A, B, 1)$ to the SRUC game $(A, B, \Dcal)$
just scales the payoff functions by a non-zero value.

Any Nash equilibrium for an RUC game using stationary strategies
is also a Nash equilibrium for the corresponding SRUC game,
so \ref{item:ne-equiv:d} implies \ref{item:ne-equiv:sd}
and \ref{item:ne-equiv:1} implies \ref{item:ne-equiv:s1}.

For $C \in \{A, B\}$, let $\rho_C \defeq e^{(C,\infty)}_1(\xstar, \ystar)$.
If $(\xstar, \ystar)$ is a Nash equilibrium for the SRUC game $(A, B, 1)$,
then by \cref{thm:meml-ne-is-eigen}, we get that
$(\rho_B, \xstar)$ is an eigenpair of $B\trnF$ and $(\rho_A, \ystar)$ is an eigenpair of $A$.
Then by \cref{thm:mult:pos-eig-gives-ne-2}, $(\xstar, \ystar)$ is a Nash equilibrium
for the RUC games $(A, B, \Dcal)$ and $(A, B, 1)$.
Hence, \ref{item:ne-equiv:s1} implies \ref{item:ne-equiv:d} and \ref{item:ne-equiv:1}.

Together, these implications among \ref{item:ne-equiv:sd}, \ref{item:ne-equiv:s1},
\ref{item:ne-equiv:d}, and \ref{item:ne-equiv:1} imply that they are all equivalent.
\end{proof}

\begin{theorem}[approximate Nash equilibrium]
\label{thm:mult:approx-ne-nmem}
Let $(\stdmeml(\xstar), \stdmeml(\ystar))$ be a Nash equilibrium for the RUC game
$(A, B, 1)$ such that $\support(\xstar) = \support(\ystar) = \toN$.
For $\delta \in [0, 1)$, $\eps \defeq 4\delta/(1-\delta)^2$,
and $\rho_C \defeq e^{(C, \infty)}(\xstar, \ystar)$ for $C \in \{A, B\}$,
let $\xhat, \yhat \in \simplex_n$ such that for all $i \in \toN$,
$|\xhat_i - \xstar_i| \le \delta\xstar_i$ and $|\yhat_i - \ystar_i| \le \delta\ystar_i$.
Then for any distribution $\Dcal$, $(\xhat, \yhat)$ is an $\eps$-approximate
Nash equilibrium for the RUC game $(A, B, \Dcal)$.
In particular, for any strategies $f$ and $g$, and any $w \ge 0$,
\begin{align*}
e^{(B,\infty)}_w(\xhat, g) &\ge \frac{1-\delta}{1+\delta}w\rho_B \ge \frac{e^{(B,\infty)}_w(\xhat, \yhat)}{1+\eps},
\\ e^{(A,\infty)}_w(f, \yhat) &\le \frac{1+\delta}{1-\delta}w\rho_A \le (1+\eps)e^{(A,\infty)}_w(\xhat, \yhat).
\end{align*}
\end{theorem}

\subsubsection{Equivalence of Strategies}

The notion of equivalence of strategies must be modified to account for collision thresholds.

\begin{definition}
Let $I \defeq [i_1, \ldots, i_k]$ and $J \defeq [j_1, \ldots, j_k]$. For any $w \ge 1$,
The pair $(I, J)$ is called $w$-\emph{collision-consistent} if $\colls(I', J') \le w-1$,
where $I' \defeq [i_1, \ldots, i_{k-1}]$ and $J' \defeq [j_1, \ldots, j_{k-1}]$.
\end{definition}

\begin{definition}
For $w \ge 1$, two randomized strategies $f_1$ and $f_2$ are said to be $w$-\emph{equivalent}
if for every $w$-collision-consistent pair $(I, J)$, we have
$\Pr(\isResponse(f_1, I, J)) = \Pr(\isResponse(f_2, I, J))$.
\end{definition}

We now show that given a pair of strategies, replacing each strategy by an equivalent
strategy makes no difference to anyone's payoff.

\begin{lemma}
\label{thm:mult:eqv-score}
Let $f_1$ and $f_2$ be equivalent strategies of the \batter{}
and $g_1$ and $g_2$ be equivalent strategies of the \bowler{}.
Then for any $C \in \mathbb{R}_{\ge 0}^{n \times n}$, any $w \in \mathbb{Z}_{\ge 0}$,
and any $r \in \mathbb{Z}_{\ge 0} \cup \{\infty\}$, we have
$e^{(C,r)}_w(f_1, g_1) = e^{(C,r)}_w(f_2, g_2)$.
\end{lemma}

A corollary to \cref{thm:mult:eqv-score} is that if $(f_1, g_1)$ is a Nash equilibrium
for the RUC game $(A, B, \Dcal)$, then so is $(f_2, g_2)$.

Next, we give a useful characterization of equivalence to stationary strategies.

\begin{lemma}
\label{thm:mult:eqv-meml}
Let $\xvec \in \simplex_n$. A randomized strategy $f$ is $w$-equivalent to $\stdmeml(\xvec)$
iff $f^{(I,J)}(\emptyset) \sampledFrom \xvec$ for every
feasible history $(I, J)$ of $f$ where $\colls(I, J) < w$.
\end{lemma}

\subsubsection{Uniqueness of Nash Equilibrium up to Equivalence}

In \cref{sec:non-meml-extra:nec}, we proved the uniqueness-up-to-equivalence
of Nash equilibria for the RUC game $(A, A, 1)$, where $A$ is irreducible.
We extend that proof to RUC games $(A, A, w)$, where $w \ge 1$.
The main ideas in the proof, and the way we break the proof into lemmas, remains the same.

Let $\rho$ be the Perron root of $A$, and
$\xstar$ and $\ystar$ be the left and right Perron vectors of $A$, respectively.
We will show that for every Nash equilibrium $(f^*, g^*)$ of the RUC game $(A, A, w)$,
$f^*$ is equivalent to $\stdmeml(\xstar)$ and $g^*$ is equivalent to $\stdmeml(\ystar)$.

For notational convenience, we write $S^{(r)}_w(f, g)$ and $e^{(r)}_w(f, g)$ instead of
$S^{(A, r)}_w(f, g)$ and $e^{(A, r)}_w(f, g)$.
If $(f^*, g^*)$ is any Nash equilibrium for the RUC game $(A, A, w)$,
then by \cref{thm:all-ne-same-value}, we get
$e^{(\infty)}_w(f^*, g^*) = e^{(\infty)}_w(f^*, \ystar) = e^{(\infty)}_w(\xstar, g^*) = w\rho$,
and that $(\xstar, g^*)$ and $(f^*, \ystar)$ are also Nash equilibria.

We first show (in \cref{thm:mult:xhat-equals-xstar,thm:mult:yhat-equals-ystar})
that if $(f^*, g^*)$ is a Nash equilibrium for the RUC game $(A, A, w)$, where $w \ge 1$,
then $f^*(\emptyset) \sampledFrom \xstar$ and $g^*(\emptyset) \sampledFrom \ystar$.

\begin{lemma}
\label{thm:mult:yhat-equals-ystar}
Let $(f^*, g^*)$ be any Nash equilibrium for the RUC game $(A, A, w)$, where $w \ge 1$.
Then $g^*(\emptyset) \sampledFrom \ystar$.
\end{lemma}

\begin{lemma}
\label{thm:mult:xhat-equals-xstar}
Let $(f^*, g^*)$ be any Nash equilibrium for the RUC game $(A, A, w)$, where $w \ge 1$.
Then $f^*(\emptyset) \sampledFrom \xstar$.
\end{lemma}

Next, we show (in \cref{thm:mult:subg-nopt,thm:mult:subf-nopt,thm:mult:sub-ne,thm:mult:rec-sub-ne}) that
residual strategies of Nash equilibria are themselves Nash equilibria.

\begin{lemma}
\label{thm:mult:subg-nopt}
Let $(f^*, g^*)$ be any Nash equilibrium for the RUC game $(A, A, w)$, where $w \ge 1$.
Then for any $i, j \in \toN$ and any strategy $f$,
$e^{(\infty)}_{w-\boolone(i=j)}(f, {g^*}^{([j],[i])}) \le (w - \boolone(i=j))\rho$.
\end{lemma}

\begin{lemma}
\label{thm:mult:subf-nopt}
Let $(f^*, g^*)$ be any Nash equilibrium for the RUC game $(A, A, w)$, where $w \ge 1$.
Then for any $i, j \in \toN$ and any strategy $g$,
$e^{(\infty)}_{w-\boolone(i=j)}({f^*}^{([i],[j])}, g) \ge (w-\boolone(i=j))\rho$.
\end{lemma}

\begin{lemma}
\label{thm:mult:sub-ne}
Let $(f^*, g^*)$ be any Nash equilibrium for the RUC game $(A, A, w)$, where $w \ge 1$.
Then for any $i, j \in \toN$, $({f^*}^{([i],[j])}, {g^*}^{([j],[i])})$ is
a Nash equilibrium for the RUC game $(A, A, w - \boolone(i=j))$. Moreover,
$e^{(\infty)}_{w-\boolone(i=j)}({f^*}^{([i],[j])}, {g^*}^{([j],[i])}) = (w-\boolone(i=j))\rho$.
\end{lemma}
\begin{proof}
Follows from \cref{thm:mult:subg-nopt,thm:mult:subf-nopt}.
\end{proof}

\begin{lemma}
\label{thm:mult:rec-sub-ne}
Let $(f^*, g^*)$ be any Nash equilibrium for the RUC game $(A, A, w)$, where $w \ge 1$.
Then for any lists $I$ and $J$ of the same length such that $\colls(I, J) < w$,
we get that $({f^*}^{(I,J)}, {g^*}^{(J,I)})$ is a Nash equilibrium
for the RUC game $(A, A, w - \colls(I, J))$.
\end{lemma}

Finally, we combine the above results to show uniqueness-up-to-equivalence of Nash equilibria.

\begin{theorem}
\label{thm:mult:nec-cond-for-ne}
For any Nash equilibrium $(f^*, g^*)$ of the RUC game $(A, A, w)$, where $w \ge 1$,
$f^*$ is $w$-equivalent to $\stdmeml(\xstar)$ and $g^*$ is $w$-equivalent to $\stdmeml(\ystar)$.
\end{theorem}
\begin{proof}
For any lists $I$ and $J$ such that $\colls(I, J) < w$, we get that
$({f^*}^{(I,J)}, {g^*}^{(J,I)})$ is a Nash equilibrium for the RUC game
$(A, A, w-\colls(I, J))$ by \cref{thm:mult:rec-sub-ne}.
By \cref{thm:mult:xhat-equals-xstar,thm:mult:yhat-equals-ystar}, we get that
${f^*}^{(I,J)}(\emptyset) \sampledFrom \xstar$ and ${g^*}^{(J,I)}(\emptyset) \sampledFrom \ystar$.
Hence, by \cref{thm:mult:eqv-meml}, $f^*$ is $w$-equivalent to $\stdmeml(\xstar)$,
and $g^*$ is $w$-equivalent to $\stdmeml(\ystar)$.
\end{proof}

\section{SRUC Games on Reducible Matrices}
\label{sec:reducible}

In this section, we investigate the SRUC game $(A, B)$,
where $A$ or $B$ may be reducible (i.e., may not be irreducible).
We show some necessary conditions on Nash equilibria
and prove that a Nash equilibrium always exists.

\subsection{Handling Trivial Cases}

When $\graph(A) \nsubseteq \graph(B)$, a Nash equilibrium exists by \cref{thm:edge-in-A-not-in-B}.
Hence, from now on we assume $\graph(A) \subseteq \graph(B)$.

\begin{lemma}
In the SRUC game $(A, B)$, suppose the $j\Th$ column of $B$ is zero.
Let $\xstar \in \simplex_n$ such that $\xstar_j > 0$.
Then $(\xstar, \evec^{(j)})$ is a Nash equilibrium.
\end{lemma}
\begin{proof}
$e^{(B)}(\xstar, \evec^{(j)}) = 0$, so the \bowler{} has no incentive to deviate.
Since $\graph(A) \subseteq \graph(B)$, $A[j,j]$ can be the only positive entry (if any)
in the $j\Th$ column of $A$. Hence, $e^{(A)}(\xstar, \evec^{(j)}) = A[j,j]$,
and the \batter{} cannot gain by deviating.
Hence, $(\xstar, \evec^{(j)})$ is a Nash equilibrium.
\end{proof}

Hence, from now on we assume that every column of $B$ has a positive entry.

\subsection{Necessary Conditions for Nash Equilibria}

We first prove some necessary conditions for Nash equilibria
to build some intuition about the properties that they have.

\begin{observation}
\label{thm:graph-score}
For $C \in \{A, B\}$, let $\graph(C) \defeq (\toN, E_C)$.
If the \batter{} and \bowler{} use stationary strategies $\xvec$ and $\yvec$, respectively,
then $\graph(C)$ tells us the following about $e^{(C)}(\xvec, \yvec)$:
\begin{enumerate}
\item Every vertex $u \in \toN$ contributes $\xvec_u\yvec_u$ to the probability of a collision in a round,
    i.e., $\xvec\trn\yvec = \sum_{u=1}^n \xvec_u\yvec_u$.
\item Every vertex $u \in \toN$ contributes $\xvec_u\yvec_uC[u,u]$ to the per-round payoff,
    and every edge $(u, v) \in E_C$ contributes $\xvec_u\yvec_vC[u,v]$, i.e.,
    \[ \xvec\trn C\yvec = \sum_{u \in V}\xvec_u\yvec_uC[u,u] + \sum_{(u,v) \in E_C}\xvec_u\yvec_vC[u,v]. \]
\end{enumerate}
\end{observation}

\begin{lemma}
\label{thm:posA-implies-posB}
For any $\xvec, \yvec \in \simplex_n$ such that $\xvec\trn\yvec = 0$, we get
$\xvec\trn A\yvec > 0 \implies \xvec\trn B\yvec > 0$.
\end{lemma}
\begin{proof}
For $C \in \{A, B\}$, let $(\toN, E_C) \defeq \graph(C)$.
Then $\xvec\trn A\yvec > 0$ and $\xvec\trn\yvec = 0$ imply that
$\xvec_u\yvec_vA[u,v] > 0$ for some $(u, v) \in E_A$.
Since $\graph(A) \subseteq \graph(B)$, we get
$\xvec_u\yvec_vB[u, v] > 0$, which implies $\xvec\trn B\yvec > 0$.
\end{proof}

\begin{lemma}
\label{thm:nasheq-pos-fin}
Let $(\xstar, \ystar)$ be a Nash equilibrium for the SRUC game $(A, B)$.
Then $e^{(B)}(\xstar, \ystar)$ is finite, and $e^{(A)}(\xstar, \ystar)$ is positive and finite.
\end{lemma}
\begin{proof}
The \bowler{} can guarantee a finite cost by playing a full-support strategy,
so $e^{(B)}(\xstar, \ystar)$ is finite.
The \batter{} can guarantee a positive score by playing a full-support strategy,
so $e^{(A)}(\xstar, \ystar) > 0$.

Suppose $e^{(A)}(\xstar, \ystar) = \infty$. Then $\xstarT\ystar = 0$
and $\xstarT A\ystar > 0$. By \cref{thm:posA-implies-posB}, we get $\xstarT B\ystar > 0$.
Hence, $e^{(B)}(\xstar, \ystar) = \infty$, which is a contradiction.
Hence, $e^{(A)}(\xstar, \ystar)$ is finite.
\end{proof}

\begin{definition}
In a digraph $(V, E)$, the set $S \subseteq V$ is called \emph{predecessor-inclusive}
if for all $(u, v) \in E$, we have $v \in S \implies u \in S$
(i.e., for every $v \in S$, its predecessors are also in $S$.)
\end{definition}

\begin{lemma}
\label{thm:ystar-pred-incl}
Let $(\xstar, \ystar)$ be a Nash equilibrium for the SRUC game $(A, B)$.
Then $\support(\ystar)$ is predecessor-inclusive in $\graph(A)$.
\end{lemma}
\begin{proof}
Let $S_y \defeq \support(\ystar)$ and $\graph(A) \defeq (\toN, E_A)$.
Suppose $S_y$ is not predecessor-inclusive in $\graph(A)$.
Then $\exists (u, v) \in E_A$ such that $u \not\in S_y$ and $v \in S_y$.

If the \batter{} plays the strategy $\evec^{(u)}$,
then a collision never happens (${\evec^{(u)}}\trn\ystar = 0$),
and the per-round score is positive (${\evec^{(u)}}\trn A\ystar \ge A[u,v]\ystar_v > 0$).
Hence, $e^{(A)}(\xstar, \ystar) = \infty$, which contradicts \cref{thm:nasheq-pos-fin}.
Hence, $S_y$ is predecessor-inclusive in $\graph(A)$.
\end{proof}

\subsection{Existence of Nash Equilibrium}

\subsubsection{SCC Decomposition of the Payoff Matrices}

By \cref{thm:ystar-pred-incl}, in any Nash equilibrium $(\xstar, \ystar)$
of the SRUC game $(A, B)$, $\support(\ystar)$ must be a union of
some of the strongly connected components (SCCs) of $\graph(A)$.

Furthermore, for any square matrix $C$, the SCCs of $\graph(C)$
correspond to irreducible submatrices of $C$.
Formally, for any matrix $C \in \mathbb{R}^{n \times n}_{\ge 0}$ and any sets $S, T \subseteq \toN$,
let $C[S,T]$ be the submatrix of $C$ formed by rows $S$ and columns $T$.
Then for any SCC $S \subseteq \toN$ of $\graph(C)$, $C[S, S]$ is irreducible.

Motivated by these observations, we find a Nash equilibrium for the SRUC game $(A, B)$
by carefully picking SCCs of $\graph(A)$ and $\graph(B)$ and then using
the results for irreducible matrices from \cref{sec:meml}.
To do this, we first need to describe the structure of $\graph(A)$ and $\graph(B)$.

In any directed graph $(V, E)$, let $S_1$ and $S_2$ be SCCs.
We say that $S_1$ \emph{precedes} $S_2$ iff there is an edge $(u, v) \in E$
such that $u \in S_1$ and $v \in S_2$.
This precedence relation between SCCs is a partial order,
and can be represented as a DAG (directed acyclic graph).
An SCC is called a \emph{source} if it is not preceded by any other SCC.

Let $\{V_1, V_2, \ldots\}$ be the SCCs of $\graph(B)$.
\WLoG{}, let $V_1, \ldots, V_q$ be the sources among these SCCs.
Let $V_{\le q} \defeq \bigcup_{i=1}^q V_i$
and $V_{>q} \defeq \toN \setminus V_{\le q}$ ($V_{>q}$ may be empty).
Since every DAG has a source, $q \ge 1$ and $V_{\le q} \neq \emptyset$.
For all $i \in \toN[q]$, let $\rho^{(i)}_B$ be the Perron root of $B[V_i, V_i]$.
\WLoG{}, assume $\rho^{(1)}_B \le \ldots \le \rho^{(q)}_B$.

Since $\graph(A)$ is a subgraph of $\graph(B)$, we get that each SCC of $\graph(A)$
is contained in some SCC of $\graph(B)$.
Furthermore, one of the source SCCs of $\graph(A)$ must be contained in $V_1$.
Let $\{U_1, U_2, \ldots\}$ be the strongly connected components (SCCs) of $\graph(A)$.
Assume \wLoG{} that $U_1$ is a source SCC and $U_1 \subseteq V_1$.
Let $\overline{U_1} \defeq \toN \setminus U_1$.

A sequence $[S_1, \ldots, S_k]$ is called a \emph{contiguous partition} of $\toN$
if it is a partition of $\toN$ and for all $i < j$,
if $u \in S_i$ and $v \in S_j$, then $u < v$.
\WLoG{}, we can renumber the actions $\toN$ so that $[U_1, \overline{U_1}]$
and $[V_1, V_2, \ldots, V_q, V_{>q}]$ are contiguous partitions.

We will construct strategies for the SRUC game $(A, B)$ by finding
strategies for \emph{subproblems} and then \emph{composing} them together.
Formally, these are the steps we take:
\begin{enumerate}
\item Pick a contiguous partition $[S_1, \ldots, S_k]$ of $\toN$.
\item Pick an appropriate strategy $\xvec^{(i)} \in \simplex_{|S_i|}$ for each $i \in \toN[k]$.
\item Pick $[\alpha_1, \ldots, \alpha_k] \in \simplex_k$.
\item Concatenate the vectors $\alpha_1\xvec^{(i)}, \alpha_2\xvec^{(2)}, \ldots, \alpha_k\xvec^{(k)}$
    to get $\xvec \in \simplex_n$. We denote this as
    $\xvec \defeq \concat([\alpha_1\xvec^{(1)}, \ldots, \alpha_k\xvec^{(k)}])$.
\end{enumerate}
We also use the opposite process, i.e.,
given $\xvec \in \simplex_n$, $[S_1, \ldots, S_k]$, and $[\alpha_1, \ldots, \alpha_k]$,
we can find the corresponding $[\xvec^{(1)}, \ldots, \xvec^{(k)}]$.

\subsubsection{\bowler{}'s strategy}

We now turn to showing existence of Nash equilibria.
We start by finding a strategy for the \bowler{} that's a candidate
for being part of a Nash equilibrium.

First, we show that if the \bowler{} plays any strategy $\ystar \in \simplex_n$ whose support is
predecessor-inclusive in $\graph(A)$, then the \batter{} cannot get an infinite score.
Furthermore, the \batter{}'s score doesn't depend on her actions outside $\support(\ystar)$,
it only depends on how her actions are distributed within $\support(\ystar)$.

\begin{lemma}
\label{thm:ystar-pred-incl-implies-0-disjoint}
Suppose $\ystar \in \simplex_n$ is predecessor-inclusive for $\graph(A)$.
Then for all $\xvec \in \simplex_n$ such that $\xvec\trn\ystar = 0$,
we get $e^{(A)}(\xvec, \ystar) = 0$.
\end{lemma}
\begin{proof}
Let $S_x \defeq \support(\xvec)$ and $S_y \defeq \support(\ystar)$. Then
$\xvec\trn A\ystar = \sum_{i \in S_x}\sum_{y \in S_y} \xvec_i\ystar_jA[i,j]$.
So, $\xvec\trn A\ystar > 0$ iff there is an edge in $\graph(A)$ going from $S_x$ to $S_y$.
Since $S_y$ is predecessor-inclusive, that is not possible, so $\xvec\trn A\ystar = 0$,
so $e^{(A)}(\xvec, \ystar) = 0$.
\end{proof}

\begin{lemma}
\label{thm:batter-can-expand}
Let $\xstar, \ystar \in \simplex_n$ such that $\xstarT\ystar > 0$
and $\ystar$ is predecessor-inclusive in $\graph(A)$.
Let $S_y \defeq \support(\ystar)$. Let $\alpha \defeq \sum_{i \in S_y} \xstar_i$.
(Then $\alpha \in (0, 1]$.) Define $\xhat \in \simplex_n$ as
\[ \xhat_i \defeq \begin{cases}\xstar_i / \alpha & \textrm{ if } i \in S_y
\\ 0 & \textrm{ if } i \not\in S_y \end{cases}. \]
Then $e^{(A)}(\xstar, \ystar) = e^{(A)}(\xhat, \ystar)$.
\end{lemma}
\begin{proof}
If $\alpha = 1$, then $\support(\xstar) \subseteq S_y$, so $\xhat = \xstar$,
and so the lemma is trivially true. Now consider the case where $\alpha \neq 1$.
Define $\xtild \in \simplex_n$ as
\[ \xtild_i \defeq \begin{cases}0 & \textrm{ if } i \in S_y
\\ \xstar_i/(1-\alpha) & \textrm{ if } i \not\in S_y \end{cases}. \]
Then $\xstar = \alpha\xhat + (1-\alpha)\xtild$.
\[ e^{(A)}(\xstar, \ystar) = \frac{\xstarT A\ystar}{\xstarT\ystar}
= \frac{\alpha(\xhat\trn A\ystar) + (1-\alpha)\xtild\trn A\ystar}{\alpha\xhat\trn A\ystar}
= e^{(A)}(\xhat, \ystar) + \frac{1-\alpha}{\alpha}\frac{\xtild\trn A\ystar}{\xhat\trn\ystar}. \]
Since $\ystar$ is predecessor-inclusive, and $\xtild\trn\ystar = 0$,
we get $\xtild\trn A\ystar = 0$ by \cref{thm:ystar-pred-incl-implies-0-disjoint}.
Hence, $e^{(A)}(\xstar, \ystar) = e^{(A)}(\xhat, \ystar)$.
\end{proof}

Since $U_1$ is an SCC of $\graph(A)$, $A[U_1, U_1]$ is irreducible.
Let $\rho_A$ be the Perron root of $A[U_1, U_1]$
and let $\vvec$ be the right Perron vector of $A[U_1, U_1]$.

\begin{lemma}
\label{thm:ystar-source-of-A}
Let $\ystar \defeq \concat([\vvec, \zerovec]) \in \simplex_n$.
Then $\forall \xvec \in \simplex_n$, $e^{(A)}(\xvec, \ystar) = 0$ if
$\xvec\trn\ystar = 0$ and $e^{(A)}(\xvec, \ystar) = \rho_A$ otherwise.
\end{lemma}
\begin{proof}
Since $\support(\ystar) = U_1$, we get that $\ystar$ is predecessor-inclusive in $\graph(A)$.
For any $\xvec \in \simplex_n$, let
$\xvec = \concat([\alpha\uvec, (1-\alpha)\wvec])$, where $\uvec \in \simplex_{|U_1|}$.
If $\alpha = 0$, then $e^{(A)}(\xvec, \ystar) = 0$ by \cref{thm:ystar-pred-incl-implies-0-disjoint}.
If $\alpha > 0$, then $e^{(A)}(\xvec, \ystar) = e^{(A[U_1,U_1])}(\uvec, \vvec) = \rho_A$
by \cref{thm:batter-can-expand}.
\end{proof}

If we can find a strategy $\xstar$ such that the \bowler{}'s expected total cost
is minimized at the $\ystar$ as defined in \cref{thm:ystar-source-of-A},
then $(\xstar, \ystar)$ would be a Nash equilibrium for the SRUC game $(A, B)$.

\subsubsection{\batter{}'s strategy}

Finding a suitable candidate for the \batter{}'s strategy
in a Nash equilibrium is more involved.
We define a strategy for the \batter{}, called a \emph{comet strategy}
and show that it restricts the \bowler{} to play strategies
supported on source SCCs of $\graph(B)$.

The key idea is that if $(u, v)$ is an edge in $\graph(B)$,
and if the \batter{} plays $u$ with much higher probability than $v$,
then $v$'s contribution towards the probability of collision
is much less than its contribution to the expected per-round score.
Hence, the \bowler{} will be disincentivized to play $v$.

Let $\eps_B$ be the smallest positive entry in $B$.
Recall that $[V_1, V_2, \ldots]$ are the SCCs of $\graph(B)$,
and the first $q$ of them are sources.
Let $\uvec^{(i)}$ be the left Perron vector of $B[V_i, V_i]$ for each $i \in \toN[q]$.
Let $\eps_u$ be the smallest entry in any $\uvec^{(i)}$, i.e.,
$\eps_u \defeq \min_{i=1}^q \min_{j=1}^{|V_i|} u^{(i)}_j$.

Run a graph traversal algorithm (like depth-first-search or breadth-first-search)
on $\graph(B)$ starting from $V_{\le q}$.
This will give us a directed forest, where each tree's root is in $V_{\le q}$
and non-root vertices are in $V_{>q}$.
For every vertex $v \in V_{>q}$, let $\pi(v)$ be the predecessor in this forest,
and $d(v)$ be the distance of $v$ from a vertex in $V_{\le q}$.

\begin{definition}[comet strategy]
Let $\delta \defeq \eps_B/2\rho^{(q)}_B$.
Let $V_{>q} \defeq \{v_1, \ldots, v_m\}$.

If $m = 0$, then let $\beta = 0$ and
$\xstar \defeq \concat([\uvec^{(1)}/n, \ldots, \uvec^{(q)}/n])$.

If $m > 0$, let $\wvec \in \mathbb{R}_{\ge 0}^m$ where $\wvec_i \defeq \delta^{d(v_i)-1}$.
Let $\beta \defeq (1+n/(\delta\eps_u))^{-1}$ and
$\xstar \defeq \concat([(1-\beta)\uvec^{(1)}/n, \ldots, (1-\beta)\uvec^{(q)}/n,
    \beta\wvec/\|\wvec\|_1])$.

Then $\xstar$ is called a comet strategy for the matrix $B$.%
\footnote{%
For those curious about the name \emph{comet},
think of $\xstar$ as a comet whose body is $V_{\le q}$ and tail is $V_{>q}$.
$\beta$ gives the size of the tail and $\delta$ tells us how sharply the tail thins out.}
\end{definition}

\begin{lemma}
\label{thm:comet-delta-tail}
Let $\xstar$ be a comet strategy for $B$.
Then for any $j \in V_{>q}$, we have $\xstar_{\pi(j)} \ge \xstar_j/\delta$.
\end{lemma}
\begin{proof}
If $d(j) \ge 2$, then $\pi(j) \in V_{>q}$.
Hence, $\xstar_{\pi(j)} = \xstar_j/\delta$.

We have $1-\beta \ge \beta n/(\delta\eps_u)$.
If $d(j) = 1$, then $\pi(j) \in V_{\le q}$. Then $\xstar_j \le \beta$ and
$\xstar_{\pi(j)} \ge (1-\beta)\eps_u/n \ge \beta/\delta \ge \xstar_j/\delta$.
\end{proof}

\begin{lemma}
\label{thm:comet-forces-rho1}
Let $\xstar$ be a comet strategy for $B$.
Let $\yhat \in \simplex_n$ such that $\support(\yhat) \subseteq V_{\le q}$.
Then $\rho^{(1)}_B \le e^{(B)}(\xstar, \yhat) \le \rho^{(q)}_B$.

If $V_{>q} \neq \emptyset$, then let $\ytild \in \simplex_n$ such that
$\support(\ytild) \subseteq V_{>q}$. Then $e^{(B)}(\xstar, \ytild) > \rho^{(q)}_B$.
Moreover, for any $0 \le \gamma \le 1$, we have
$e^{(B)}(\xstar, \gamma\yhat + (1-\gamma)\ytild) \ge e^{(B)}(\xstar, \yhat)$.
\end{lemma}
\begin{proof}
Let $\yhat = \concat([b_1\vvec^{(1)}, \ldots, b_q\vvec^{(q)}, \zerovec])$,
where $[b_1, \ldots, b_q] \in \simplex_q$
and $\vvec^{(i)} \in \simplex_{|V_i|}$ for all $i \in \toN[q]$.
By \cref{thm:sum-ratio-min-max},
\begin{align*}
e^{(B)}(\xstar, \yhat)
&= \frac{\sum_{i=1}^q (\uvec^{(i)}/n)\trn B[V_i,V_i](b_i\vvec^{(i)})}{%
    \sum_{i=1}^q (\uvec^{(i)}/n)\trn(b_i\vvec^{(i)})}
\\ &\in \left[\min_{i=1}^q \frac{{\uvec^{(i)}}\trn B[V_i,V_i]\vvec^{(i)}}{{\uvec^{(i)}}\trn\vvec^{(i)}},
    \max_{i=1}^q \frac{{\uvec^{(i)}}\trn B[V_i,V_i]\vvec^{(i)}}{{\uvec^{(i)}}\trn\vvec^{(i)}}\right]
= [\rho^{(1)}_B, \rho^{(q)}_B].
\end{align*}
Using \cref{thm:comet-delta-tail}, we get
\begin{align*}
\xstarT B\ytild &= \sum_{i=1}^n\sum_{j=1}^n \xstar_i\ytild_jB[i,j]
\ge \sum_{j \in V_{>q}}\xstar_{\pi(j)}\ytild_jB[i,j]
\ge \frac{\eps_B}{\delta}\sum_{j \in V_{>q}}\xstar_j\ytild_j
\ge \frac{\eps_B}{\delta}\xstarT\ytild.
\end{align*}
Hence,
\[ e^{(B)}(\xstar, \ytild) = \frac{\xstarT B\ytild}{\xstarT\ytild}
\ge \frac{\eps_B}{\delta} > \rho^{(q)}_B. \]
Using \cref{thm:sum-ratio-min-max}, we get
\[ e^{(B)}(\xstar, \gamma\yhat + (1-\gamma)\ytild)
= \frac{\gamma\xstarT B\yhat + (1-\gamma)\xstarT B\ytild}{\gamma\xstarT\yhat + (1-\gamma)\xstarT\ytild}
\ge e^{(B)}(\xstar, \yhat). \qedhere \]
\end{proof}

\subsubsection{Existence of Nash Equilibrium}

We now show that the candidate strategies we considered indeed give us a Nash equilibrium.
Recall that $\vvec$ is the right Perron vector of $A[U_1,U_1]$,
and $\rho_A$ is the Perron root of $A[U_1,U_1]$.

\begin{lemma}
Let $\ystar \defeq \concat([\vvec, \zerovec])$
and let $\xstar$ be a comet strategy for $B$.
Then $(\xstar, \ystar)$ is a Nash equilibrium for the SRUC game $(A, B)$.
Moreover, $e^{(A)}(\xstar, \ystar) = \rho_A$ and $e^{(B)}(\xstar, \ystar) = \rho_B$.
\end{lemma}
\begin{proof}
Since $\support(\xstar) = \toN$,
by \cref{thm:ystar-full-supp}, the \batter{} cannot gain by deviating.

Since $\support(\ystar) \subseteq V_1$, let
$\ystar = \concat([\zvec, 0])$, where $\zvec \in \simplex_{|V_1|}$.
Then
\[ e^{(B)}(\xstar, \ystar)
= \frac{(\uvec^{(1)}/n)\trn B[V_1,V_1]\zvec}{(\uvec^{(1)}/n)\trn\zvec} = \rho^{(1)}_B. \]
By \cref{thm:comet-forces-rho1}, for any $\yvec \in \simplex_n$,
$e^{(B)}(\xstar, \yvec) \ge \rho^{(1)}_B$.
Hence, the \bowler{} cannot gain by deviating.
Hence, $(\xstar, \ystar)$ is a Nash equilibrium.
\end{proof}

We could have defined comet strategy a little differently:
instead of $\delta \defeq \eps_B/2\rho^{(q)}_B$, we could have picked any
$0 < \delta < \eps_B/\rho^{(q)}_B$,
and instead of giving the same weight of $(1-\beta)/n$ to each $\uvec^{(i)}$,
we could have distributed the weight of $1-\beta$ differently.
This doesn't make any difference iff $B$ is irreducible.
Hence, when $B$ is reducible, we get that multiple Nash equilibria exist.


\begin{thebibliography}{10}

\bibitem{gtappfinance}
Franklin Allen and Stephen Morris.
\newblock Finance applications of game theory.
\newblock Center for financial institutions working papers, Wharton School
  Center for Financial Institutions, University of Pennsylvania, 1998.
\newblock URL: \url{https://EconPapers.repec.org/RePEc:wop:pennin:98-23}.

\bibitem{AlonVempala}
Noga Alon, Troy Lee, Adi Shraibman, and Santosh Vempala.
\newblock The approximate rank of a matrix and its algorithmic applications:
  Approximate rank.
\newblock In {\em Symposium on Theory of Computing (STOC)}, pages 675--684.
  ACM, 2013.
\newblock \href {https://doi.org/10.1145/2488608.2488694}
  {\path{doi:10.1145/2488608.2488694}}.

\bibitem{hc-ios-app}
Shubham Arya.
\newblock Hand cricket ({iOS} app).
\newblock URL: \url{https://theshubhamarya.github.io/HandCricket/}.

\bibitem{Aumann1994}
Robert~J. Aumann and Lloyd~S. Shapley.
\newblock {\em Long-Term Competition---A Game-Theoretic Analysis}, pages 1--15.
\newblock Springer New York, 1994.
\newblock \href {https://doi.org/10.1007/978-1-4612-2648-2_1}
  {\path{doi:10.1007/978-1-4612-2648-2_1}}.

\bibitem{Barman15}
Siddharth Barman.
\newblock Approximating nash equilibria and dense bipartite subgraphs via an
  approximate version of caratheodory's theorem.
\newblock In {\em Symposium on Theory of Computing (STOC)}, pages 361--369.
  ACM, 2015.
\newblock \href {https://doi.org/10.1145/2746539.2746566}
  {\path{doi:10.1145/2746539.2746566}}.

\bibitem{finitelyRepeatedGames}
Jean-Pierre Benoit and Vijay Krishna.
\newblock Finitely repeated games.
\newblock {\em Econometrica}, 53(4):905--922, 1985.
\newblock \href {https://doi.org/10.2307/1912660} {\path{doi:10.2307/1912660}}.

\bibitem{borie09pursuitevasion}
Richard Borie, Craig Tovey, and Sven Koenig.
\newblock Algorithms and complexity results for pursuit-evasion problems.
\newblock In {\em International Joint Conference on Artificial Intelligence
  (IJCAI)}, page 59–66, 2009.
\newblock URL: \url{https://www.ijcai.org/Proceedings/09/Papers/021.pdf}.

\bibitem{gtappwireless}
Dimitris~E. Charilas and Athanasios~D. Panagopoulos.
\newblock A survey on game theory applications in wireless networks.
\newblock {\em Computer Networks}, 54(18):3421--3430, 2010.
\newblock \href {https://doi.org/10.1016/j.comnet.2010.06.020}
  {\path{doi:10.1016/j.comnet.2010.06.020}}.

\bibitem{chen06NEapproxsmooth}
Xi~Chen, Xiaotie Deng, and Shang{-}Hua Teng.
\newblock Computing nash equilibria: Approximation and smoothed complexity,
  2006.
\newblock \href {https://arxiv.org/abs/cs/0602043} {\path{arXiv:cs/0602043}}.

\bibitem{Chen06Sparse}
Xi~Chen, Xiaotie Deng, and Shang-Hua Teng.
\newblock Sparse games are hard.
\newblock In {\em Workshop on Internet and Network Economics (WINE)}, pages
  262--273. Springer, 2006.

\bibitem{chen07approxwinlose}
Xi~Chen and Shang-hua Teng.
\newblock The approximation complexity of win-lose games.
\newblock In {\em Symposium on Discrete Algorithms (SODA)}, page 159–168.
  SIAM, 2007.

\bibitem{dask05Threenash}
Konstantinos Daskalakis and Christos Papadimitriou.
\newblock Three-player games are hard.
\newblock {\em Electronic Colloquium on Computational Complexity (ECCC)}, 2005.

\bibitem{hc-instructables}
dissknight.
\newblock How to play hand cricket.
\newblock URL: \url{https://www.instructables.com/How-to-Play-Hand-Cricket/}.

\bibitem{friedman1971}
James~W. Friedman.
\newblock A non-cooperative equilibrium for supergames.
\newblock {\em The Review of Economic Studies}, 38(1):1--12, 1971.
\newblock \href {https://doi.org/10.2307/2296617} {\path{doi:10.2307/2296617}}.

\bibitem{fudenbergmaskin}
Drew Fudenberg and Eric Maskin.
\newblock The folk theorem in repeated games with discounting or with
  incomplete information.
\newblock {\em Econometrica}, 54(3):533--554, 1986.
\newblock \href {https://doi.org/10.2307/1911307} {\path{doi:10.2307/1911307}}.

\bibitem{kannanTheobald}
Ravi Kannan and Thorsten Theobald.
\newblock Games of fixed rank: A hierarchy of bimatrix games.
\newblock In {\em Symposium on Discrete Algorithms (SODA)}, pages 1124--1132.
  SIAM, 2007.
\newblock URL: \url{http://dl.acm.org/citation.cfm?id=1283383.1283504}.

\bibitem{khan2007game}
Mohammad~Emtiyaz Khan.
\newblock Game theory models for pursuit evasion games, 2007.
\newblock URL: \url{https://emtiyaz.github.io/Writings/EMTgame.pdf}.

\bibitem{gtapprobots}
S.M. LaValle and S.~Hutchinson.
\newblock Game theory as a unifying structure for a variety of robot tasks.
\newblock In {\em International Symposium on Intelligent Control}, pages
  429--434. IEEE, 1993.
\newblock \href {https://doi.org/10.1109/ISIC.1993.397675}
  {\path{doi:10.1109/ISIC.1993.397675}}.

\bibitem{mehta14const}
Ruta Mehta.
\newblock Constant rank bimatrix games are ppad-hard.
\newblock In {\em Symposium on Theory of Computing (STOC)}, pages 545--554.
  ACM, 2014.
\newblock \href {https://doi.org/10.1145/2591796.2591835}
  {\path{doi:10.1145/2591796.2591835}}.

\bibitem{repeatedGamesBook}
Jean-François Mertens, Sylvain Sorin, and Shmuel Zamir.
\newblock {\em Repeated Games}.
\newblock Cambridge University Press, 2015.
\newblock \href {https://doi.org/10.1017/CBO9781139343275}
  {\path{doi:10.1017/CBO9781139343275}}.

\bibitem{perronfrobenius}
Carl~D. Meyer and Ian Stewart.
\newblock {\em Matrix Analysis and Applied Linear Algebra, Second Edition}.
\newblock SIAM, 2023.
\newblock \href {https://doi.org/10.1137/1.9781611977448}
  {\path{doi:10.1137/1.9781611977448}}.

\bibitem{poweriteration}
R.~V. Mises and H.~Pollaczek-Geiringer.
\newblock Praktische verfahren der gleichungsauflösung.
\newblock {\em ZAMM - Journal of Applied Mathematics and Mechanics /
  Zeitschrift für Angewandte Mathematik und Mechanik}, 9(2):152--164, 1929.
\newblock \href {https://doi.org/10.1002/zamm.19290090206}
  {\path{doi:10.1002/zamm.19290090206}}.

\bibitem{nash1951NE}
John Nash.
\newblock Non-cooperative games.
\newblock {\em Annals of mathematics}, pages 286--295, 1951.
\newblock \href {https://doi.org/10.2307/1969529} {\path{doi:10.2307/1969529}}.

\bibitem{vonNeumann1928}
J.~von Neumann.
\newblock Zur theorie der gesellschaftsspiele.
\newblock {\em Mathematische Annalen}, 100:295--320, 1928.
\newblock URL: \url{http://eudml.org/doc/159291}.

\bibitem{agtbook}
Noam Nisan, Tim Roughgarden, {\'E}va Tardos, and Vijay~V. Vazirani, editors.
\newblock {\em Algorithmic Game Theory}.
\newblock Cambridge University Press, 2007.
\newblock \href {https://doi.org/10.1017/CBO9780511800481}
  {\path{doi:10.1017/CBO9780511800481}}.

\bibitem{osborneGameTheoryTextBook}
Martin~J. Osborne and Ariel Rubinstein.
\newblock {\em A Course in Game Theory}.
\newblock The MIT Press, 1994.
\newblock URL: \url{https://ideas.repec.org/b/mtp/titles/0262650401.html}.

\bibitem{papa94PPAD}
Christos~H. Papadimitriou.
\newblock On the complexity of the parity argument and other inefficient proofs
  of existence.
\newblock {\em J. Comput. Syst. Sci.}, 48(3):498--532, 1994.
\newblock \href {https://doi.org/10.1016/S0022-0000(05)80063-7}
  {\path{doi:10.1016/S0022-0000(05)80063-7}}.

\bibitem{rubinsteininfinitegames}
A.~Rubinstein.
\newblock Strong perfect equilibrium in supergames.
\newblock {\em International Journal of Game Theory}, 9(1):1--12, 1980.
\newblock \href {https://doi.org/10.1007/BF01784792}
  {\path{doi:10.1007/BF01784792}}.

\bibitem{Rubinstein1993hideseek}
Ariel Rubinstein and A.~Tversky.
\newblock Naive strategies in zero-sum games.
\newblock Working papers, Tel Aviv - the Sackler Institute of Economic Studies,
  1993.
\newblock URL: \url{https://EconPapers.repec.org/RePEc:fth:teavsa:17-93}.

\bibitem{Rubinstein1997hideseek}
Ariel Rubinstein, Amos Tversky, and Dana Heller.
\newblock {\em Naive Strategies in Competitive Games}, pages 394--402.
\newblock Springer Berlin Heidelberg, 1997.
\newblock \href {https://doi.org/10.1007/978-3-642-60495-9_30}
  {\path{doi:10.1007/978-3-642-60495-9_30}}.

\bibitem{shapley1953stochastic}
L.~S. Shapley.
\newblock Stochastic games.
\newblock {\em Proceedings of the National Academy of Sciences},
  39(10):1095--1100, 1953.
\newblock \href {https://doi.org/10.1073/pnas.39.10.1095}
  {\path{doi:10.1073/pnas.39.10.1095}}.

\bibitem{zerosumrepeatedgames}
Sylvain Sorin.
\newblock {\em A First Course on Zero-Sum Repeated Games}.
\newblock Springer, 2003.

\end{thebibliography}
\end{document}